\documentclass[11pt]{article}
\usepackage{graphicx} %
\usepackage[utf8]{inputenc}
\usepackage[margin=1in]{geometry}
\usepackage{amsmath,amsthm,amssymb}
\usepackage{xcolor}
\usepackage{bm}
\usepackage[hyphens]{url}
\usepackage{hyperref}
\usepackage[nameinlink,capitalise]{cleveref}
\usepackage{bbold}
\usepackage{mathtools}
\usepackage{doi}
\usepackage{tikz}
\usepackage{thm-restate}
\usepackage[ruled]{algorithm2e}
\usepackage[T1]{fontenc}
\usepackage{indentfirst}
\usepackage{libertine}
\usepackage{natbib}
\usepackage{xspace}

\usepackage{etoolbox}
\usepackage{todonotes}
\newtoggle{DEBUG}
\togglefalse{DEBUG}
\newtoggle{COMMENTS}
\toggletrue{COMMENTS}

\usepackage[showdeletions]{color-edits} 
\addauthor[Sid]{sb}{blue}
\addauthor[Kamesh]{km}{magenta}
\addauthor[Kangning]{kn}{orange}

\usetikzlibrary{arrows.meta, positioning}

\newcommand{\E}{\mathbb{E}}
\newcommand{\R}{\mathbb{R}}

\newcommand{\bu}{\bm{u}}
\newcommand{\bv}{\bm{v}}
\newcommand{\bem}{\bm{m}}
\newcommand{\bxi}{\bm{\xi}}

\newcommand{\hp}{\hat{p}}
\newcommand{\SW}{\mathrm{SW}}
\newcommand{\NE}{\mathrm{NE}}

\newcommand{\cut}{\mathrm{cut}}
\newcommand{\ber}{\mathtt{Bernoulli}}

\newcommand{\nz}{\mathcal{Z}}
\newcommand{\bnz}{\bm{\mathcal{Z}}}
\newcommand{\bpis}{\bm{\pi}_{*}}
\newcommand{\nI}{\mathcal{I}}
\newcommand{\nf}{\mathcal{F}}
\newcommand{\nt}{\mathcal{T}}
\newcommand{\nn}{\mathcal{N}}
\newcommand{\bnf}{\bm{\mathcal{F}}}

\newcommand{\PoA}{\mathrm{PoA}}
\newcommand{\nw}{\mathcal{W}}

\newcommand{\bfz}{\bm{\mathsf{Z}}}
\newcommand{\bsi}{\bm{\sigma}}

\newcommand{\bmpi}{\bm{\pi}}
\newcommand{\PSi}{\bigotimes_{i = 1}^N \Sigma_i}

\newcommand{\principal}{principal\xspace}

\newtheorem{theorem}{Theorem}[section]
\newtheorem{lemma}[theorem]{Lemma}

\theoremstyle{definition}
\newtheorem{definition}[theorem]{Definition}
\newtheorem{example}[theorem]{Example}

\renewcommand{\d}{\mathop{\mathrm{d}\!}}

\crefname{Program}{Program}{Programs}
\creflabelformat{Program}{(#2\textup{#1})#3}

\renewcommand{\epsilon}{\varepsilon}

\makeatletter
\def\@fnsymbol#1{\ensuremath{\ifcase#1\or \dagger\or \ddagger\or
   \mathsection\or \mathparagraph\or \|\or **\or \dagger\dagger
   \or \ddagger\ddagger \else\@ctrerr\fi}}
\makeatother
\newcommand*\samethanks[1][\value{footnote}]{\footnotemark[#1]}

\definecolor{linkc}{rgb}{0.1, 0.5, 0.7}
\definecolor{citec}{rgb}{0.6, 0.3, 0.7}
\definecolor{urlc}{rgb}{0.5, 0.1, 0.2}
\hypersetup{
    colorlinks=true,
    linkcolor=linkc,
    citecolor=citec,
    urlcolor=urlc
}

\title{The Price of Competitive Information Disclosure}
\author{Siddhartha Banerjee\thanks{School of Operations Research and Information Engineering, Cornell University, Ithaca, NY 14853. Email: \textsf{sbanerjee@cornell.edu}. Supported by NSF grants ECCS-1847393, CNS-195599 and AFOSR grant FA9550-23-1-0068.} \and Kamesh Munagala\thanks{Department of Computer Science, Duke University, Durham, NC 27708-0129. Emails: \textsf{kamesh@cs.duke.edu}, \textsf{yiheng.shen@duke.edu}. Supported by NSF grant IIS-2402823.} \and Yiheng Shen\samethanks[2] \and Kangning Wang\thanks{Department of Computer Science, Rutgers University, Piscataway, NJ 08854. Email: \textsf{kn.w@rutgers.edu}.}}
\date{}

\begin{document}

\maketitle
\begin{abstract}
In many decision-making scenarios, individuals strategically choose what information to disclose to optimize their own outcomes. It is unclear whether such strategic information disclosure can lead to good societal outcomes. To address this question, we consider a competitive Bayesian persuasion model in which multiple agents selectively disclose information about their qualities to a principal, who aims to choose the candidates with the highest qualities. Using the price-of-anarchy framework, we quantify the inefficiency of such strategic disclosure. We show that the price of anarchy is at most a constant when the agents have independent quality distributions, even if their utility functions are heterogeneous. This result provides the first theoretical guarantee on the limits of inefficiency in Bayesian persuasion with competitive information disclosure.
\end{abstract}

\thispagestyle{empty}

\newpage
\tableofcontents

\thispagestyle{empty}

\newpage
\setcounter{page}{1}

\section{Introduction}

Every day, people make critical decisions based on limited information, but often \emph{someone else} controls what information they see. For example, when job seekers apply to a company, they highlight their best qualities (such as a prestigious degree or a successful project) and downplay less flattering information. The employer then must decide who to hire based on the r\'{e}sum\'{e}s and references provided. Similarly, in college admissions, students choose which recommendation letters to send or whether to submit optional test scores, and hence curate the pictures that admissions officers see. Finally, as a very recent application~\citep*{collina2025emergentalignmentcompetition}, consider an entity that is querying strategic AI agents and using their responses to optimize its action. In all examples, the party presenting the information (for instance, the AI agent or applicant) has an incentive to reveal the information selectively to improve their chances of being selected. Such strategic revelation of information is not lying; it is about being selective in revealing the truth. 

From the perspective of someone who aims to optimize societal outcomes, agents revealing information should only be helpful, since the societal optimizer always has the option of ignoring the information revealed. However, quantitatively, to what extent can information revelation improve societal outcomes? We first look at a simple example.
\begin{example}
\label{eg:intro}
Consider a scenario where there are $N$ agents and each of them has a private quality drawn independently from a $\ber(1/N)$ distribution.\footnote{This means each agent's quality is $1$ with probability $1/N$, and is $0$ with probability $1 - 1/N$.} The quality of an agent represents the value that she can bring to the society. The principal can select up to one agent. If the principal has no additional information, his best strategy is to arbitrarily select an agent to obtain an expected quality of $1/N$. On the other hand, in the ``first-best'' scenario where the principal has access to full information, he can obtain an expected quality of $1-(1-1/N)^N \approx 1-1/e$. The ratio between these values is $\Theta(N)$, showing that accessing agents' value information helps improve the societal outcome significantly.
\end{example}

\paragraph{Voluntary information disclosure and the price of concealment.} If the agents are forced to reveal their qualities, then the principal always gets the first best. However, it may not be practical to assume full information revelation, in light of philosophical and legal considerations and privacy concerns. At the other extreme, if agents can arbitrarily self-report their qualities as in a classical model known as ``cheap talk'', then nothing prevents an agent from misreporting the highest possible quality. The model of \emph{Bayesian persuasion}~\citep*{kamenica2011bayesian} provides a suitable middle ground, where an agent can meaningfully persuade the principal by ex ante \emph{committing} to a signaling scheme, which maps the quality to signals before realizing the quality of the agent. The next example illustrates possible interactions of agents who engage in Bayesian persuasion. 
\begin{example}[continued from \cref{eg:intro}]
In the previous setting, one option for all agents is to commit to revealing their true qualities. However, these commitments do not form an equilibrium: If agents $\{2,3,\ldots,N\}$ follow this policy, then agent $1$ is better off not revealing anything, since she can then win with a probability of $(1-1/N)^{N-1} \approx 1/e \gg 1/N$. In a different scenario where agents $\{2,3,\ldots,N\}$ decide not to disclose anything, agent $1$ can do much better by sending a signal $\sigma$ with a small probability of $\frac{\varepsilon}{1 - 1 / N}$ conditioned on her value being $0$, and otherwise send signal $\sigma'$. In this case, the mean of  $\sigma'$ is $ \frac{\Pr[\textrm{her value is $1$}]}{1 - \Pr[\sigma \textrm{ is sent}]} = \frac{1/N} {1 -  \frac{\varepsilon}{1 - 1 / N} \cdot (1 - 1 / N)} = \frac{1}{(1 - \varepsilon)N}$. Since $\frac{1}{(1 - \varepsilon)N} > 1/N$, she will be selected as long as $\sigma'$ is sent, with a probability of $1 - \varepsilon \gg 1/N$.
\end{example}

The idea of \emph{strategic information revelation} formalizes voluntary information disclosure in competitive settings and provides a framework for reasoning about the aforementioned real-world scenarios. In this work, we consider a model of competitive Bayesian persuasion first studied by \citet*{boleslavsky2015grading}, and subsequently by \citet*{AuKawai} and by \citet*{Du2024} among others. In our model, there are multiple agents with private information about their \emph{qualities}. Those qualities are drawn independently from commonly known priors, which are not necessarily identical. 
A single principal (or \emph{receiver}) can select up to $k$ agents and aims to maximize the expected sum of their qualities. Each agent also derives some utility from being selected based on a given monotone function of her quality. Each agent can disclose information about her quality to the principal via some \emph{signaling scheme}, which is a probabilistic mapping from qualities to signals. Upon receiving the signals from each agent, the principal updates his belief on each agent's quality, and finally selects the agents with the highest posterior means of quality. 

In our model, we assume agents receive utility on selection that is an arbitrary monotone function of their quality. Such monotone utility functions are naturally motivated. As an example, the agents may only care about the event that they are selected, and derive a fixed, short-term utility from selection. This would correspond to the agents' utility function being a constant independent of their quality. On the other hand, the long-term reward an agent may receive (for instance, pay or grades) subsequent to being selected may depend on their quality. If the agent's utility depends on such long-term rewards, it can be modeled as a monotone function of quality. Indeed, different agents can prioritize short-term versus long-term rewards differently, leading to heterogeneity in their utility functions.

This signaling model formalizes our main goal of characterizing the extent to which competitive information disclosure can improve societal outcomes. As in our example, a natural benchmark for the principal is the ``first best''---the sum of qualities of selected agents when the principal has full information of agents' qualities. As a crucial constraint in Bayesian persuasion models, each agent's posterior distribution conditioned on their signal must be \emph{Bayes plausible}, i.e., the signaling scheme must be a valid probabilistic mapping from their qualities to signals. Since agents compete with each other, their overall strategies must form a \emph{Nash equilibrium}, wherein any agent's signaling scheme is a best response to all other agents' strategies based on her own utility function for being selected. Finally, adopting the lens of the price of anarchy (PoA), we define the notion of \emph{price of concealment} for our setting: \emph{How much worse can an equilibrium outcome be, compared to the first-best outcome?} Here, we define the price of concealment as the ratio in the principal's objective (expected sum of qualities of the selected agents) between the full-information outcome and the strategic outcome.

\begin{example}[continued from~\cref{eg:intro}]
For $N$ agents with qualities drawn from i.i.d.\@ Bernoulli distributions $\ber(1/N)$, suppose that each agent gets a constant utility of $1$ from being selected. Then, following the results of \citet*{AuKawai} (see also \cref{thm:lb}), there is a \emph{unique} symmetric Nash equilibrium. In this equilibrium, every agent implements a signaling scheme where the posterior mean of the signals forms a distribution with the cumulative distribution function $G(s)=s^{\frac{1}{N-1}}$. Moreover, the price of concealment is about $2 - 2/e \approx 1.26$ when $N$ is large.  
\end{example}
As our running example illustrates, the equilibrium strategies in simple settings can still be somewhat involved. However, in this example, the price of concealment is much smaller than $\Theta(N)$. That is to say, voluntary information disclosure can lead to a dramatic societal improvement here---but how far does this improvement extend?

\subsection{Summary of Results}

The main contribution of our work (\cref{thm:main}) is to show that the price of concealment of strategic information revelation is a constant (at most $22.19$), even when the agents have different utilities for being selected, the prior distributions over the agents' qualities are heterogeneous, and the principal can select multiple agents. We also show a lower bound of $2$ on the price of concealment (\cref{thm:lb}) even in the simple case where the agents' priors are i.i.d.\@ and Bernoulli. 

\paragraph{Remarks.} Before proceeding further, three remarks are in order. First, we define the price of concealment using the principal's welfare. This definition is somewhat different from the well-known price of anarchy (PoA) definition, as we do not consider the agents' utility (which can be an arbitrary monotone function of their qualities) in the calculation.\footnote{Some of our positive results, in particular, \cref{thm:warmup}, easily extend directly to show the same bound on PoA when we consider the utility of all agents and the principal.} This is because in the settings we consider, the principal's utility---the qualities of the selected agents---matches better with social welfare. In hiring scenarios for example, the agents' utilities might be derived from receiving monetary payments from the principal, and hence cancel out with the principal's payments when calculating social welfare. In other scenarios, agents' utilities need not be on the same scale as each other's or the principal's, and might be negligible when compared to the principal's utility. From now on, we will use the term PoA to refer to our notion (where we consider the gap in the principal's utility), in keeping with the existing literature.

Second, we highlight two fundamental modeling choices that our main result relies on. First, we assume that the agents' qualities are non-negative, which is natural in settings like hiring or admissions: candidates can be stronger or weaker, but are generally not ``harmful'' to the principal. As we show in \cref{sec:lb}, the PoA can become unbounded if qualities are allowed to be negative. The non-negativity also implies that the principal always selects as many agents as allowed. Second, we assume the agents' utility functions are monotone in their qualities. This assumption ensures that agents' incentives are not in conflict with the principal's objective, and deviating toward a signal with more probability mass on higher values will generally increase the agent's utility. Such property is critical in our utility gain analysis for the deviation signal in \cref{lem:deviation_wp_lb}.

Finally, the work of \citet*{AuKawai} and \citet*{Du2024} establishes the existence of an equilibrium in the canonical setting where agents derive constant utility from being selected, and have identical and non-identical priors respectively.\footnote{This setting extends the special case presented in \cref{sec:warmup} to non-identical priors.}  Though the existence of a Nash Equilibrium is an open question for the general scenario with heterogeneous utility functions that we consider\footnote{The equilibrium need not exist for priors with point masses~\citep*{Du2024}, but it is open whether it exists when the priors are continuous.}, our constant PoA bound directly applies to the constant utility setting. To circumvent the possible lack of an equilibrium in the general game, we present two natural variants of the game in the appendix, where a Nash equilibrium is guaranteed to exist and the constant PoA bound for the original game applies with minor changes in analysis. 

\begin{itemize}
    \item In \cref{app:discretized}, we present a discretized variant of the game where the number of signaling schemes, viewed as pure strategies, is finite. Here, under general utilities, we establish existence of Nash equilibria, and show that our PoA results gracefully degrade in the discretization parameter.

    \item In \cref{app:noisy}, we introduce a variant where the principal observes the signals with a small random noise, so that their selection is noisy. Under smoothness assumptions on the noise, we show a Nash equilibrium exists, and for sufficiently small noise, our PoA bounds carry over.
\end{itemize}  
Since games are typically implemented in practice via finite menus with discretized reports, and the receiver usually cannot observe signals perfectly, these results indicate that our PoA analysis can be extended to these more realistic, yet general settings, where an equilibrium is guaranteed to exist.

\paragraph{Techniques.} In contrast with problems for which well-trodden approaches directly show PoA bounds (for instance, the smoothness framework~\citep*{smoothness} or analysis via potential functions~\citep*{PoS}), there are two aspects in which our problem is novel. First, the strategy space of any agent is the set of signaling schemes, which is infinite in size even for discrete quality distributions. Second and more importantly, the utilities are asymmetric between the agents and the principal, so that the agents are responding to the principal's optimization objective using their own utility functions. Furthermore, the social welfare is not the sum of all agents' utilities, but rather, only the principal's utility. At an abstract level, this makes our problem similar to PoA for revenue in auctions, as the revenue (rather than the welfare) captures the seller's utility and is unrelated to buyers' utility. Although in any Bayes-Nash equilibrium of an auction, the revenue can be viewed as ``virtual welfare'' via Myerson's lemma, which still enables the application of the smoothness framework~\citep*{HartlineHT14}, such a characterization is not available in our model.

Similarly, in Bayesian persuasion, the optimization behavior of the principal is fixed and cannot be controlled by the agents. This makes the welfare maximization problem, even in the absence of strategic behavior, non-convex, and in fact, we do not know constant-factor approximation algorithms for this setting even when the principal chooses one agent~\citep*{BanerjeeM0025}. This fact makes it difficult to give a characterization in terms of potential games, and these complexities in the strategic Bayesian persuasion problem necessitate the development of a novel analytic framework.

Our framework is inspired by techniques from stochastic optimization problems, particularly prophet inequalities~\citep*{KrengelS} and the core--tail decomposition technique in auction design \citep*{BabaioffILW}. We develop related techniques to decompose an agent's distribution and explicitly construct a signaling scheme that they can use to deviate and improve their utility if the PoA of the equilibrium is too high. This construction is intricate and represents our main technical contribution. As a warm-up to our main result and to present the main ideas of core--tail decomposition and signal construction cleanly, in \cref{sec:warmup}, we consider the special case where all agents have identical priors and derive constant utility from being selected. We show a PoA of $4$ for this case and highlight the challenges in generalizing this result.

\subsection{Related Work}
\paragraph{Strategic information revelation and Bayesian persuasion.} Our work is closely related to the field of information design, which studies how information is selectively disclosed to influence decision-making (see surveys by~\citet*{bergemann2019information,dughmi2017algorithmic}). A foundational model in this area is Bayesian persuasion \citep*{kamenica2011bayesian}, where a well-informed sender (in our case, the agent) strategically designs signals to influence the actions of a receiver (in our case, the principal) who seeks to maximize their own utility. This framework has been applied in diverse domains, including price discrimination \citep*{AlijaniBMW22,bergemann2015limits,Banerjee2024fair,cummings2020algorithmic,RoeslerS}, security games \citep*{DBLP:conf/aaai/XuRDT15}, AI alignment~\citep*{collina2025emergentalignmentcompetition}, and regret minimization \citep*{DBLP:conf/sigecom/BabichenkoTXZ21}. As mentioned before, the problem has significant non-convexity since the sender does not control the receiver's optimization routine. Nevertheless,  \citet*{dughmi2016algorithmic} show computationally efficient algorithms for finding signaling schemes that optimize the sender's utility. These algorithms do not extend to our setting where there are multiple agents with independent priors. As such, our setting is already challenging even without the strategic aspect, and \citet*{BanerjeeM0025} present logarithmic approximation algorithms for this setting when the receiver chooses one agent.

\paragraph{Competitive Bayesian persuasion games.}
Unlike classical Bayesian persuasion, where a central mediator controls information flow, our setting features competitive information revelation, where multiple self-interested agents selectively disclose signals about themselves to a decision-maker (receiver or principal). As noted before, an agent's strategic behavior is in the selection of the signaling scheme, which we assume is Bayes plausible. Special cases of this Bayesian persuasion game have been widely studied in various settings. \citet*{boleslavsky2015grading} analyze how schools strategically set grading standards to influence an evaluator's perceptions, while \citet*{boleslavsky2018limited} consider project selection under limited capacity, where agents compete by producing persuasive evidence. Similarly,~\citet*{jain2019competing} examine competition among senders targeting a rationally inattentive receiver, while \citet*{au2021competitive} extend these ideas by incorporating correlated information among senders. \citet*{koessler2022long} study a dynamic variant of information disclosure where two competing senders repeatedly release public signals to a single receiver over an unbounded time horizon. They show that disclosure equilibrium is one-shot if there is no constraint on the senders' disclosure policies. \citet*{banerjee2021threshold} study the price of anarchy in a signaling game, where the agents are restricted to use only threshold tests as signaling strategies. All these models consider only the competitive behavior between two agents.

For multiple agents, \citet*{hwang2019competitive} study equilibria with competitive information revelation in advertising and pricing, while \citet*{hoffmann2020persuasion} investigate the implications of selective disclosure for marketing and privacy regulation. \citet*{gradwohl2022reaping} consider the scenario when the receiver can only obtain signals from a single agent. More recently, \citet*{ding2023competitive} propose a model for competitive information design in a Pandora’s Box framework, and \citet*{sapiro2024persuasion} analyze persuasion when receiver preferences are ambiguous. \citet*{tang2024intrinsic} consider a Bayesian persuasion version of prophet inequalities, where the receiver sets thresholds based on known priors of the agents -- this converts the game to a two-agent setting, as now each sender is competing only against the receiver, and not other senders. Relatedly, there are other sequential persuasion models inspired by classical stopping problems. \citet*{hahn2020prophet} study Bayesian persuasion through the lens of prophet inequalities, and \citet*{hahn2022secretary} study the secretary recommendation problem in a similar framework. These works focus on sequential arrivals and stopping rules, whereas our work utilizes the prophet-inequality style analysis on Price of Anarchy in a multi-agent competitive environment. Finally, a different line of work studies the persuasion game when agents share a common state, and reveal information about this state to the receiver~\citep*{gentzkow2016competition,gentzkow2017bayesian,ravindran2020competing,hossain2024multi}. 

In contrast with these models, we study the Bayesian persuasion game with multiple senders and a value-maximizer receiver. The existence of an equilibrium in this model has been shown with symmetric~\citep*{AuKawai} and asymmetric agents~\citep*{Du2024}. In their models, every agent gains a constant utility when selected; on the other hand, we study a more general setting where the utilities of each agent can be heterogeneous.  Further, our work focuses on the welfare loss arising from strategic signaling via the price of anarchy. As far as we are aware, we are the first to consider the notion of price of anarchy for general Bayesian persuasion games.

\paragraph{Fairness and strategic selection.} Strategic information revelation is particularly important in selection settings, such as hiring and college admissions, where decision-makers evaluate candidates based on self-reported attributes. Traditional selection policies prioritize meritocratic ranking based on observable scores, but these scores can be biased due to demographic correlations~\citep*{dixon2013race}.  A complementary line of research models uncertainty in evaluation measures and proposes randomized selection mechanisms to ensure fairness~\citep*{emelianov2020fair,mehrotra2021mitigating,garciasoriano2021maxmin,singh2021fairness,shen2023fairness,devic2024stability,BanerjeeM0025}. Our work extends this perspective by incorporating strategic information disclosure into the selection process. Rather than assuming a fixed information structure, we study how self-interested agents selectively reveal information to persuade decision-makers, and analyze the resulting impact on social efficiency.

\section{Preliminaries}
\label{sec:prelims}

We consider a competitive information disclosure setting with $N$ agents (\emph{senders}) and a single principal (\emph{receiver}). Each agent has a private quality type (or \emph{value}) drawn from some known prior distribution, and can disclose some partial information regarding this type to the principal via some chosen Bayesian persuasion (or \emph{signaling}) scheme. The principal aims to select a specified number of agents to maximize the sum of selected values; on the other hand, each agent receives some utility from being selected.

\subsection{Agent Values and Signaling Schemes} \label{sec:prelim_signals}

Formally, our setting is parameterized by the number of allowed selections $k$, as well as an \emph{instance}
\[
\nI = (\bnf, \bu),
\]
where $\bnf = (\nf_1, \nf_2, \ldots, \nf_{N})$ denotes the agents' prior value distributions and $\bu = (u_1, \ldots, u_N)$ denotes the agents' utility functions. For each agent $i \in [N]$, her value $v_i$ is drawn independently from the distribution $\nf_i \in \Delta([0, 1])$ over the value space $[0, 1]$.\footnote{Although we assume that the value space is $[0, 1]$ for simplicity, our results hold for any compact subset $\mathcal{V} \subseteq \mathbb{R}_{\ge 0}$ by normalizing the values.}\footnote{Even though we believe that assuming non-negative values on agents is natural (since hiring someone is usually considered at least not ``harmful'' to a company), we provide an example at the end of \cref{app:negative_example} showing that the price of anarchy can be unbounded when agents' values are allowed to be negative.}
For simplicity, we condition on the instance $\nI$ throughout and omit its explicit mention in later definitions. 

For each agent $i$, her strategy is a signaling scheme $\nz_i = (\Sigma_i, \pi_i)$, where $\Sigma_i$ is some measurable space of \emph{signals}, and $\pi_i \colon [0,1] \to \Delta(\Sigma_i)$ is a mapping from a realized value $v$ to a probability measure over $\Sigma_i$.\footnote{Formally, we assume that the function $\pi_i$ is a probability kernel from $[0,1]$ to the signal space $\Sigma_i$, that is, for every measurable set $S \subseteq \Sigma_i$, the mapping $v \mapsto \pi_i(S \mid v)$ is measurable.} Agent $i$ commits to signaling scheme $\nz_i$ before realizing her value; subsequently, after learning $v_i$, she samples a signal $\sigma \in \Sigma_i$ to send to the principal according to the distribution $\pi_i(\cdot \mid v_i)$. 

Given signaling scheme $\nz_i$, we denote the induced distribution of agent $i$'s signal by $\Pi_i \in \Delta(\Sigma_i)$. Formally, for any measurable signal set $S \subseteq \Sigma_i$, the distribution $\Pi_i(S)$ is given by 
\[
\Pi_i(S) = \int_{[0,1]} \pi_i(S \mid v) \d \nf_i(v).
\]

The priors $\{\nf_i\}_{i \in [N]}$ and signaling schemes $\{\nz_i\}_{i \in [N]}$ are common knowledge to the principal and all agents. Let $\bnz = (\nz_1, \ldots, \nz_N)$ denote a strategy profile and denote by $\sigma_i$ the signal sent by agent $i$. Upon receiving $\sigma_i$ from agent $i$, the principal can update his belief about $v_i$ via Bayes' Rule, and compute the posterior mean $e_i(\sigma_i)$ as
\[
e_i(\sigma_i) = \E[v_i \mid \sigma_i] = \frac{\int_{[0,1]} v \cdot \pi_i(\d \sigma_i \mid v) \d \nf_i(v)}{\int_{[0,1]} \pi_i(\d \sigma_i \mid v) \d \nf_i(v)}. \footnote{The notation $\pi_i(\cdot \mid v_i)$ is a probability measure over the signal space $\Sigma_i$. Note that ``$\d\sigma_i$'' represents an infinitesimal element of space from which the signal $\sigma_i$ is drawn. This notation allows us to seamlessly handle signaling schemes that are continuous or discrete. For example, if $\pi_i(\cdot \mid v)$ is continuous and has a density function $p_i(\sigma_i \mid v)$, then the notation $\pi_i( \d \sigma_i \mid v)$ is equivalent to $p_i(\sigma_i \mid v) \d \sigma_i$. On the other hand, if the agent sends one of a finite set of signals, the measure $\pi_i(S \mid v)$ is a sum of probability masses of signals in $S$. Then the notation $\pi_i(\d \sigma_i \mid v)$ simply refers to the probability mass on the signal $\sigma_i$.}
\]
We emphasize our assumption that agent $i$'s strategic behavior is only in the selection of the scheme $\nz_i$. Once this scheme is fixed, the agent signals her value truthfully according to this scheme. This ability to commit to a scheme is a main feature of \emph{Bayesian persuasion}~\citep{kamenica2011bayesian} and is what allows agents to meaningfully affect the principal's decision (i.e., makes the signals go beyond ``cheap talk'').

\subsection{Individual Utilities and Social Welfare} \label{sec:prelim_utilities_and_SW}

If agent $i \in [N]$ is selected with value $v$, she receives utility $u_i(v)$. We assume that the function $u_i\colon [0,1] \to \mathbb{R}_{\ge 0}$ is continuous and monotonically increasing, with the property that $u_i(v) > 0$ for any $v > 0$.\footnote{Otherwise, if $u_i(v) \equiv 0$ for all agents, the PoA can become trivially unbounded, since any strategy profile will be a Nash equilibrium.}

The aim of the principal is to maximize the sum of the values of selected agents. Consequently, given the signals, he selects the $k$ agents with the highest posterior mean, breaking ties at random. Formally, for a signal profile $\bm{\sigma} = (\sigma_1, \ldots, \sigma_N) \in \bigotimes_{i = 1}^N \Sigma_i$, we can sort the posterior means $e_i(\sigma_i)$ in decreasing order
\[
e_{(1)}(\bm{\sigma}) \ge e_{(2)}(\bm{\sigma}) \ge \cdots \ge e_{(N)}(\bm{\sigma}),
\]
and define the social welfare as the sum of the $k$ highest posterior means,\footnote{More generally, our model allows the principal to get value $w(v)$ by selecting an agent with realized value $v$, for any strictly increasing function $w\colon [0,1] \to \mathbb{R}_{\ge 0}$ (by appropriately transforming the value distributions). } i.e., $
\SW(\bm{\sigma}) = \sum_{j = 1}^{k} e_{(j)}(\bm{\sigma})$. 

We define the \emph{first-best social welfare} (i.e., the expected maximum social welfare achievable when the principal observes all agents' values) as follows.
\begin{definition}[First-Best Social Welfare]
\label{def:bestwelfare}
Given instance $\nI = (\bnf, \bu)$, the \emph{first-best (social) welfare} on $\nI$ is
\[
\SW^{\max}(\nI) = \E_{(v_1, v_2, \ldots, v_N) \sim \bigotimes_{i = 1}^N \nf_i} \left[\max_{S \subseteq [N],\ |S| = k} \sum_{i \in S} v_i \right].
\]
\end{definition}

Denote the agents selected without tying with others by $\nw(\bm{\sigma})$, and denote agent(s) with posterior mean equal to $e_{(k)}(\bm{\sigma})$ (i.e., those who possibly form a tie and get selected) by $\nt(\bm{\sigma})$. We have
\[
\nw(\bm{\sigma}) = \{i \in [N] : e_i(\sigma_i) > e_{(k)}(\bm{\sigma})\}; \quad \nt(\bm{\sigma}) = \{i \in [N] : e_i(\sigma_i) = e_{(k)}(\bm{\sigma})\}.
\]
Denote agent $i$'s probability of being selected on the signal profile $\bm{\sigma}$ by $\rho_i(\bm{\sigma})$. Since the principal breaks ties at random, $\rho_i(\bm{\sigma})$ is given by:
\[
\rho_i(\bm{\sigma}) = \begin{cases}
    1 & \text{if $i \in \nw(\bm{\sigma})$,} \\
    \frac{k - |\nw(\bm{\sigma})|}{|\nt(\bm{\sigma})|} & \text{if $i \in \nt(\bm{\sigma})$,} \\
    0 & \text{otherwise}.
\end{cases}
\]

The utility of agent $i$ on the signal profile $\bm{\sigma}$ is given by
\[
u_i(\bm{\sigma}) = \rho_i(\bm{\sigma}) \cdot \E[u_i(v_i) \mid \sigma_i] = \rho_i(\bm{\sigma}) \cdot \frac{\int_{[0,1]} u_i(v) \cdot \pi_i(\d \sigma_i \mid v) \d \nf_i(v)}{\int_{[0,1]} \pi_i(\d \sigma_i \mid v) \d \nf_i(v)}.
\]

\subsection{Nash Equilibrium and Price of Anarchy}

We study the game when all agents' strategies form a Nash equilibrium with respect to the instance $\nI = (\bnf, \bu)$. A strategy profile is a Nash equilibrium if for every agent $i$, her strategy yields the highest utility given all others' strategies fixed. In other words, she cannot benefit by deviating to any other strategy. We now formally define this using the notation introduced above.

First, under any strategy profile $\bnz = (\nz_1, \ldots, \nz_N)$, the realized signal profile $\bm{\sigma} = (\sigma_1, \ldots, \sigma_N)$ is sampled from a joint distribution given by the product measure  
\[
\bm{\Pi} = \Pi_1 \times \Pi_2 \times \cdots \times \Pi_N.
\]
Moreover, under strategy profile $\bm{\nz}$, the expected utility of each agent $i$, as well as the expected social welfare, are given by
\begin{align*}
u_i(\bm{\nz}) &= \E_{\bm{\sigma} \sim \bm{\Pi}}\left[u_i(\bm{\sigma})\right] = \int_{\bigotimes_{i = 1}^N \Sigma_i} \rho_i(\bm{\sigma}) \cdot \E[u_i(v_i) \mid \sigma_i] \d  \bm{\Pi}(\bm{\sigma});\\
\SW(\bm{\nz}) &= \E_{\bm{\sigma} \sim \bm{\Pi}}\left[\SW(\bm{\sigma})\right] =  \E_{\bm{\sigma} \sim \bm{\Pi}}\left[\sum_{j = 1}^{k} e_{(j)}(\bm{\sigma}) \right] = \int_{\bigotimes_{i = 1}^N \Sigma_i}\left(\sum_{j = 1}^{k} e_{(j)}(\bm{\sigma}) \right) \d \bm{\Pi}(\bm{\sigma}).
\end{align*}
Using this notation, we can have the following definition.
\begin{definition}[Nash Equilibrium Signaling]
Given an instance $\nI = (\bnf, \bu)$, a strategy profile \[\bm{\nz} = (\nz_1, \nz_2, \ldots, \nz_N)\] is a Nash equilibrium on $\nI$ if and only if for every agent $i \in [N]$ and any alternative signaling scheme $\nz_i'$, we have 
\[
u_i(\nz_1, \ldots, \nz_{i - 1}, \nz_i, \nz_{i + 1}, \ldots, \nz_N) \ge u_i(\nz_1, \ldots, \nz_{i - 1}, \nz_i', \nz_{i + 1}, \ldots, \nz_N).
\]
\end{definition}

Finally, we can define the price of anarchy ($\PoA$) as the ratio between the first-best welfare (i.e., the welfare when the principal observes every agent's true value), and the worst-case expected welfare attained under any Nash equilibrium. For any instance $\nI$, let $\bfz_\NE(\nI)$ denote the set of all Nash equilibrium strategy profiles on $\nI$. Then we can formally define the $\PoA$ as follows.
\begin{definition}[Price of Anarchy]
Given instance $\nI = (\bnf,\bu)$ where $\bfz_\NE(\nI) \neq \varnothing$, 
the \emph{price of anarchy} (\emph{PoA}) for $\nI$ is defined as
\[
\PoA(\nI) = \frac{\SW^{\max}(\nI)}{\inf_{\bm{\nz} \in \bfz_{\NE}(\nI)} \SW(\bm{\nz}, \nI)}
\]
where $\SW(\bm{\nz}, \nI)$ is the expected social welfare under the strategy profile $\bnz$ for instance $\nI$.
\end{definition}

\section{Lower Bounds on the Price of Anarchy}
\label{sec:lb}

In this section, we establish lower bounds on the price of anarchy (PoA) for our problem. We do so by explicitly evaluating $\PoA(\nI)$ under a simple setting, where the \principal selects only one agent (i.e., $k = 1$), and with $N$ symmetric agents with {\em Bernoulli values}---where for every $i\in [N]$, $\nf_i$ is an i.i.d.\@ $\ber(\zeta)$ distribution---and {\em constant utility of selection}---where all agents get utility $u_i(0) = u_i(1) = 1$ for $v \in \{0, 1\}$ when selected. 
This setting is considered by \citet{AuKawai}, who show that in this case, a symmetric Nash equilibrium exists, and is unique.

\begin{theorem}
\label{thm:lb}
Consider a setting with $k = 1$, symmetric agents with $\nf_i\sim\ber(\zeta)$ and $N\zeta \leq  1$, and constant utility of selection for all $i \in [N]$. Then the price of anarchy satisfies
\begin{equation*}
\PoA(\nI) \geq \left(2 - \frac{1}{N}\right)\left(\frac{1-e^{-N\zeta}}{N\zeta}\right) \in \left[1.264, \ 2 - O\left(\frac{1}{N}\right)\right].
\end{equation*}
\end{theorem}

\begin{proof}
Let $\hp = \zeta \cdot N \leq 1$ by assumption. We first note that the first-best social welfare for this setting is the probability that at least one agent has value $1$, i.e.,
\begin{equation}
\label{eq:lowerboundfb}
\SW^{\max}(\nI) = 1 - (1 - \zeta)^N \geq 1 - e^{-\hp}.
\end{equation}

Turning to signaling, observe that if two signals $\sigma_{i, 1}$ and $\sigma_{i, 2}$ share the same posterior mean, combining these into a single signal affects neither agent $i$'s utility nor welfare. Therefore, without loss of generality, we can assume that there is at most one signal for each posterior mean value. For convenience, we use a signal's posterior mean to represent the signal, thus $\Sigma_i \subseteq [0,1]$ for all agents $i$, and for any $\sigma_i \in \Sigma_i$, it holds that $e_i(\sigma_i) = \sigma_i$. The signaling scheme $\nz_i$ can therefore be represented by a posterior mean distribution (or equivalently, induced signal distribution) $\Pi_i \in \Delta([0, 1])$.
 
For agent $i$, denote the cumulative distribution function (cdf) of strategy $\Pi_i$ by $
G_i(v) = \Pr_{x \sim \Pi_i}[x \le v]$ for $v \in [0, 1]$. By the work of \citet[Section 3.1]{AuKawai}, since $\hp = \zeta \cdot N < 1$, the symmetric Nash equilibrium (that is, all $N$ agents play the same strategy) exists, and in this equilibrium, each agent's strategy $\Pi_i = \Pi$ has the same cumulative distribution $G_i = G$ as follows:
\begin{equation*}
    G(v) =
    \begin{cases}
    \left( \frac{v}{\hp} \right)^{\frac{1}{N-1}} & \quad \text{when } v \in [0, \hp], \\
    1 & \quad \text{when } v \in (\hp, 1]. 
\end{cases}
\end{equation*}
Denote the strategy profile where everyone plays the above strategy by $\bnz$. 

Let $\rho_i(v; G)$ denote the probability that agent $i$ is selected conditioned on sending a signal (with mean) $v$ and choosing the strategies of all other agents according to $G$. Then, it follows that
\begin{equation*}
    \rho_i(v; G) = \left(G(v)\right)^{N-1} = 
    \begin{cases}
        \frac{v}{\hp} & \quad \text{when } v \in [0, \hp], \\
        1 & \quad \text{when } v \in (\hp, 1].
    \end{cases}
\end{equation*}
The expected social welfare contributed from each agent $i$ is given by 
\begin{align*}
    \SW_i(\bnz) = \int_{0}^{\hp} \rho_i(v; G) \cdot v \d G(v) & = \int_{0}^{\hp}  \frac{v}{\hp} \cdot v \cdot \left ( \frac{v}{\hp} \right)^{\frac{1}{N-1} - 1} \cdot \frac{1}{\hp \cdot (N - 1)} \d v \\
     & = \int_{0}^{\hp} \left ( \frac{v}{\hp} \right)^{\frac{N}{N-1}} \cdot \frac{1}{N - 1} \d v \\ 
     & = \frac{\hp}{2N - 1}. %
\end{align*}
Since the equilibrium is symmetric, the overall social welfare is $\SW(\bnz) = N \cdot \SW_i(\bnz) = \frac{N \hat{p}}{2N - 1}$. 
Therefore, combining with~\cref{eq:lowerboundfb}, we get
\[\PoA(\nI) = \frac{\SW^{\max}(\nI)}{\SW(\bnz)} \geq \frac{(2N - 1)}{N\hp} \cdot \left(1 - e^{-\hp}\right) = \left(2-\frac{1}{N}\right)\cdot\left(\frac{1-e^{-\hp}}{\hp}\right). 
\]
This completes the proof of \cref{thm:lb}.
\end{proof}
By \cref{thm:lb}, when $\zeta$ is close to $0$ (so that $\hat{p} = N \zeta$ is small), the lower bound on the price of anarchy approaches $2$.

\paragraph{Unbounded PoA for negative-valued agents.} \label{app:negative_example}
We can construct an example where the PoA is unbounded when agents' values can be negative from the lower bound example above. Suppose we shift every agent's value distribution to the left by $\SW_i(\mathcal{Z}) - \epsilon$. Every agent using the same signaling scheme remains in equilibrium since their utilities only depend on whether they are selected. Therefore, the expected welfare in the equilibrium is $\epsilon$. The $\SW_{\max}$ is also decreased by the same gap amount, but will be much larger than $\epsilon$ (since originally there is a constant gap between these two welfare values). This leads to an unbounded PoA when the values can be negative.

\section{Warm-up: Identical Priors and Constant Utility of Selection}
\label{sec:warmup}
In this section, as a warm-up, we present a constant upper bound on PoA when all agents have identical priors, each agent gets constant utility from being selected, and the number of agents selected is $k = 1$. The ideas developed here form the basis for the more elaborate analysis of general priors and utility functions in the next section.  We show the following theorem.

\begin{theorem}
\label{thm:warmup}
When all agents have identical priors, each agent receives a constant utility $u_i(v)= 1$ for every $v \in [0,1]$ when selected, and the number of agents selected is $k = 1$, then the PoA is at most $4$.
\end{theorem}

\subsection{Proof of \texorpdfstring{\cref{thm:warmup}}{Theorem~\ref{thm:warmup}}}
Fix any instance $\nI$ (and we omit the dependence on $\nI$ in expressions for brevity). Since every agent has the same prior, we denote the common distribution as $\nf$. Let $v(q)$ denote the $q$-th quantile of $\nf$; that is, $\Pr[x \le v(q)] = q$ for $x \sim \nf$. We also denote the conditional mean of the values in the top $q$ quantile of each agent by 
\[
E_{(q)} = \frac{1}{q} \cdot \int_{1 - q}^1 v(x) \d x.
\]
Note that by definition, $q \cdot E_{(q)}$ is non-decreasing in $q \in (0,1]$.

Since the \principal can allocate a total selection probability of $1$ among all agents, an upper bound of social welfare is achieved in the scenario when the \principal assigns the entire probability to the top $1/N$ quantile of each agent's distribution. Therefore, the optimal expected social welfare satisfies
\[
\SW^{\max} \le \sum_{i \in [N]} E_{(1 / N)} \cdot (1 / N) = E_{(1 / N)}.
\]

Now consider any Nash Equilibrium $\bnz$, and let $M$ be the random variable representing the highest posterior mean among all agents. Since the \principal selects the agent with the highest posterior mean, we have $\SW(\bnz) = \E[M]$. The technical crux of our argument, formalized below in~\cref{lem:const_utility}, lies in showing that $\Pr[M \ge E_{(2 / N)}] \ge 1/2$. Combining these, we get 
\begin{align*}
    \SW(\bnz) = \E[M] &\ge E_{(2 / N)} \cdot 1/2 
     \ge \frac{E_{(2 / N)} \cdot (2 / N)}{2 / N} \cdot \frac12 
     \ge \frac{E_{(1 / N)} \cdot (1 / N)}{2 / N} \cdot \frac12 
     \ge \frac{\SW^{\max}}{4}. 
\end{align*}
This completes the proof of \cref{thm:warmup}.

\begin{lemma} \label{lem:const_utility}
Under any equilibrium $\bnz$, let $M$ denote the highest posterior mean among all agents. Then $$\Pr[M < E_{(2 / N)}] \le 1 / 2.$$
\end{lemma}
\begin{proof}
Fix equilibrium $\bnz$, and let $r_i(\bnz)$ be the overall probability of agent $i$ being selected under the equilibrium strategy profile $\bnz$, i.e., 
\[
r_i(\bnz) = \int_{\PSi} \rho_i(\bsi) \d \bm{\Pi}(\bsi).
\]
Since each agent gets utility $1$ when selected, $r_i(\bnz)$ is also the expected utility of agent $i$ under $\bnz$. 
Moreover, since the \principal can allocate a total  selection probability of $1$ across all $N$ agents, we have $\sum_{i \in [N]} r_i(\bnz) = 1$.
Thus, there exists at least one agent, say $i^*$, receiving utility $r_{i^*}(\bnz) \le 1 / N$.

Let $M_{-i^*}$ denote the highest posterior mean among all agents except $i^*$. 
We now prove the result by contradiction, by showing that if $\Pr[M < E_{(2 / N)}] > 1 / 2$, then agent $i^*$'s best response to $\{\nz_j\}_{j\neq i^*}$ gets utility strictly higher than $1/N$. 
Note that in this case, since the random variable $M_{-i^*}$ is stochastically dominated by $M$, we have 
    \begin{equation} \label{eqn:const_utility_deviation}
        \Pr[M_{-i^*} < E_{(2 / N)}] \geq \Pr[M < E_{(2 / N)}] > 1 /2.        
    \end{equation}

Now consider a potential deviating strategy by agent $i^*$, where instead of $\nz_i$, she uses a binary signaling scheme $\nz_{i^*}' = (\{s_1, s_2\}, \pi_{i^*}')$ defined as follows:
\begin{itemize}
    \item [1.] When $i^*$'s value falls in the top $2/N$ quantile of $\nf$, she sends signal $s_1$.
    \item [2.] Otherwise, she sends $s_2$. 
\end{itemize}
Formally, for any $S \subseteq \{s_1, s_2\}$,  $\pi_{i^*}'(S \mid v)$, let
\begin{equation*}
    \pi_{i^*}'(S \mid v) = 
    \begin{cases}
         \mathbb{1} [s_1 \in S] & \text{if $v$ is in the top $(2 / N)$ quantile of $\nf$,}\\
         \mathbb{1} [s_2 \in S] & \text{otherwise.\footnotemark}
      \end{cases}
\end{equation*}\footnotetext{For simplicity, we omit the case when the top $2/N$ quantile splits a mass point of $\nf$. In that case, we can split the value mass and let the higher quantile part map to $s_1$ and the lower part map to $s_2$.}%

Note that by construction of $\pi_{i^*}'$, the posterior mean of agent $i^*$'s value under signal $s_1$ is $E_{(2 / N)}$. Thus, whenever $M_{-i^*}<E_{(2/N)}$ and agent $i^*$ sends $s_1$, she is selected by the \principal. By \cref{eqn:const_utility_deviation}, conditioned on sending signal $s_1$, agent $i^*$ is selected with probability strictly greater than $1 / 2$. Moreover, since $s_1$ is sent with probability $2 / N$, agent $i^*$'s expected utility by deviating to $\nz_{i^*}'$ is at least 
\[
\Pr[\text{$s_1$ is sent}] \cdot \Pr[\text{$i^*$ is selected} \mid \text{$s_1$ is sent}] > \frac{1}{2} \cdot \frac{2}{N} = \frac{1}{N}.
\]
Since we chose $i^*$ such that $r_{i^*}(\bnz) \le 1/N$, this deviation contradicts that $\bnz$ is a Nash equilibrium.
\end{proof}

\subsection{Discussion}
The above proof critically relies on the two key properties of this setting---identical priors and constant utilities. We now discuss challenges that arise when extending to the general case with non-identical priors, heterogeneous utilities, and $k$-selection, and how the techniques introduced here can be adapted.

A central idea in our analysis is the use of \emph{quantile cuts} to capture the upper tail of each agent's value. In the warm-up case, since all agents share the same prior and $k = 1$, we apply a uniform $\frac1N$ quantile cut to each agent's distribution. This directly yields an upper bound on the maximum social welfare as the conditional mean of the values in the top $\frac1N$ quantile.  When agents have non-identical priors, a natural idea is that we can define a personalized quantile for each agent. Though the cut may vary across the agents, the core idea remains the same: We focus on the upper tail to obtain an upper bound of the optimal social welfare. 

However, using a similar approach requires more elaboration. Since now we must deal with heterogeneity between agents and $k$ can be larger than $1$, it may be impossible to construct an equal-mean cut among all agents like the warm-up case. We need to carefully handle the edge-case agents whose values are too large compared to others since the sweeping line may exhaust their entire distribution. Moreover, this method also makes the social welfare cap a combination of these heterogeneous tails, which adds more challenges to our analysis.

Another important idea is the \emph{deviation signal}. In the warm-up, if an agent deviates by reporting a signal formed by her top quantile, she can secure a higher expected payoff if the equilibrium was very inefficient. This deviation is then used to show that the equilibrium welfare must be within a constant factor of the optimum. In the general case below, the deviation strategy becomes more intricate, yet the underlying principle remains the same: By appropriately aggregating the high-value portion of an agent’s information to form a specific signal, we argue that some agent benefits in the scenario when the equilibrium is very inefficient.

Finally, in the warm-up case, all agents' utilities when selected are exactly $1$, so the benefit derived from the outcome is simply the probability of being selected. When the utilities are heterogeneous, and we make an agent deviate to a signal that is mapped to by the top certain quantile of values, then arguing about the new utility of this agent becomes challenging. The improvement in probability of being selected may not translate directly into a better utility. Therefore, we need a more careful analysis on selecting the deviating agent as well as her deviation strategy.

\section{Constant Factor Upper Bound on PoA}
\label{sec:general}
We now show the following theorem.

\begin{theorem}
\label{thm:main}
The PoA of the Bayesian persuasion game with $k$ selection (for any $k \ge 1$) is at most $11 + 5 \sqrt{5}$ when the priors $\nf_i$ are arbitrary independent distributions, and the functions $u_i(\cdot)$ are positive and weakly monotone.
\end{theorem}

The rest of this section is devoted to proving this theorem. We do so in a sequence of steps that transforms the signaling space of an individual agent, and uses the ``no deviation'' condition in equilibrium to argue about her quality. Our main step is the construction of a deviating signal $s^*$ for a specific agent, and arguing the PoA based on $s^*$ not helping the agent.

For simplicity, we will assume that we are working on a given instance $\nI$ and omit the dependence on $\nI$ in our notations in this section. Our analysis will involve parameters $\alpha, \beta, \tau$ and $\phi$ that we will optimize later. The parameters must satisfy the following constraints: 
\begin{gather*}
    \alpha, \beta, \tau, \phi > 0, \qquad
    \tau < 1, \qquad 
    \alpha > 1, \qquad \beta < \tau, \qquad \phi \cdot \tau \ge \alpha.
\end{gather*}

\subsection{Top Quantile Cuts}  \label{sec:quantile_cuts}
For any $q \in [0, 1]$, let $v_i(q)$ denote the $q$-th quantile of $\nf_i$. For simplicity, we assume that $\nf_i$ contains no point masses, and we will address in footnotes when the argument is different when point masses occur. Our first step is to perform a line sweep of the distributions $\nf_i$ so that the total quantiles in the sweep sum up to $k$. To do so, we first choose a value $E_\cut$, and define agent $i$'s quantile cut
\[
q_i = \sup \left\{x \in [0, 1] : \int_{1-x}^1{v_i(q) \d q} \ge E_\cut \cdot x \right\}.
\]
In other words, $q_i$ is the largest fraction $x$ such that the conditional mean of the top $x$ quantile is at least $E_\text{cut}$. The threshold $E_\text{cut}$ is chosen to satisfy the total quantile cut constraint:\footnote{Note that if $E_{\cut}$ equals the highest value of some $\nf_i$ and $\nf_i$ has a point mass at that value, satisfying \cref{eqn:quantile_cuts} may become impossible. A slight increase in $E_{\cut}$ can cause the corresponding agent's $q_i$ to drop abruptly to $0$. In such cases, we define $E_{\cut}$ as $\sup\{x \in [0, 1] : \sum_{i \in [N]} q_i \ge k \}$ and then reduce those $q_i$ values on the largest value mass until $
\sum_{i \in [N]} q_i = k$.}
\begin{equation} \label{eqn:quantile_cuts}
\sum_{i \in [N]} q_i = k.    
\end{equation}

For agent $i$, denote by $E_i = \E[v_i(q) \mid q > 1 - q_i] = \frac{\int_{1 - q_i}^1 v_i(q) \d q}{q_i}$ her conditioned value mean above the quantile cut $q_i$.\footnote{If $q_i = 0$, we simply set $E_i = E_\cut$.} By definition, we have $E_i \ge E_\cut$ for all $i \in [N]$. We must note that there may be agents for whom $q_i = 1$. For such agents, their conditional mean above $q_i$ (i.e., $E_i$) can differ from the common threshold $E_{\cut}$, because the line stops moving before the sweep process ends. See \cref{fig:line_sweep} for an illustration. 

\begin{figure} [htbp]
    \centering
    \includegraphics[scale = 0.9]{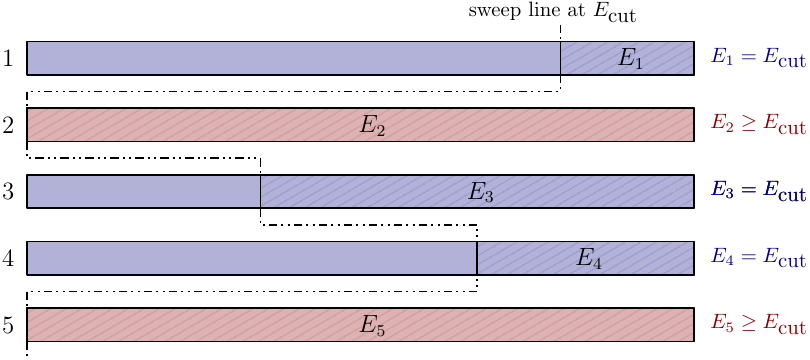}
    \caption{\em An illustration of the final state of the line sweep for an example with $5$ agents. As shown in the figure, agent $2$ and agent $5$ have a mean above the sweep line larger than $E_\cut$, as their quantile cuts satisfy $q_i = 1$. }
    \label{fig:line_sweep}
\end{figure}

Now consider the scenario where the \principal knows everyone's exact value, and selects all agents $i$ whose value is above their own quantile cut $q_i$, and additionally selects $k$ agents with highest values among the remaining agents. The social welfare in this scenario is at most
\[\SW' \le E_\cut \cdot k + \sum_{i \in [N]} E_i \cdot q_i.\]  

Since in this scenario the \principal always selects a (weak) superset of the agents selected under full-revelation of values (where the \principal knows the revealed value of all agents and selects the $k$ agents with top values), the social welfare achieved under full revelation of values $\SW^{\max}$ is at most $\SW'$.

\begin{lemma} \label{lem:swub}
    $\SW^{\max} \le E_{\cut} \cdot k + \sum_{i \in [N]} E_i \cdot q_i$.
\end{lemma}

Consider any Nash equilibrium strategy profile $\bnz$. Denote the probability of agent $i$ being selected by 
\[
r_i(\bnz) = \int_{\PSi} \rho_i(\bsi) \d \bm{\Pi}(\bsi).
\]

Fix a constant $\alpha > 1$. We group the agents into two sets $\nn_1$ and $\nn_2$ based on their $r_i(\bnz)$ as:
\[
\nn_1 = \{i \in [N] : r_i(\bnz) \ge \alpha \cdot q_i\}; \qquad \nn_2 = \{i \in [N]: r_i(\bnz) < \alpha \cdot q_i\}.
\]

Since $\sum_{i \in [N]} r_i(\bnz) \le k$, we have $\sum_{i \in \nn_1} q_i \le k/\alpha$. Therefore, $\sum_{i \in \nn_2} q_i \ge (1 - 1 / \alpha) \cdot k$. Since $1 / \alpha < 1$, $\nn_2$ cannot be an empty set. We have the following simple lemma.

\begin{lemma}
    \label{lem:N2}
    If $E_i > E_\cut$, then $i \in \nn_2$. Therefore, $\min_{i \in \nn_2} E_i \ge \max_{i \in \nn_1} E_i$.
\end{lemma}
\begin{proof}
   When $E_i > E_\cut$, we have $q_i = 1$. Since $\alpha > 1$, this implies $\alpha \cdot q_i > 1$. However, $r_i(\bnz) \le 1$, which means $i \in \nn_2$. Since $E_j \ge E_\cut$ for all $j$, it must hold that $\max_{i \in \nn_1} E_i = E_{\cut}$, completing the proof. 
\end{proof}

In the rest of the analysis, we will fix a constant $\beta > 0$. We split the analysis based on the value generated by agents in $\nn_2$ when they are selected. 
For each agent $i$, define the \emph{value contribution} of $i$ as
\[
    c_i(\bnz) = \int_{\bigotimes_{i = 1}^N \Sigma_i} \rho_i(\bsi) \cdot e_i(\sigma_i) \d \bm{\Pi}(\bsi),
\]
where $e_i(\sigma_i)$ is the posterior mean when agent $i$ sends signal $\sigma_i$. 

Notice that the overall social welfare in equilibrium $\bnz$ is \[
\SW(\bnz) = \sum_{i \in [N]} c_i(\bnz).
\]

\subsection{Case 1: All Agents in \texorpdfstring{$\nn_2$}{N-2} are Large Contributors}

We prove that when all the agents in $\nn_2$ have large value contribution, the social welfare of $\bnz$ is at least a constant fraction of the optimal social welfare.

\begin{lemma} \label{lem:large_contributors}
    Given a strategy profile $\bnz$, if all agents in $\nn_2$ satisfy $c_i(\bnz) \ge E_i \cdot \beta \cdot q_i$, then we have 
    \[
        \SW(\bnz) \ge \SW' \cdot \frac{\beta \cdot (1 - 1 / \alpha)}{2}.
    \]
\end{lemma}

\begin{proof}
    We have
    \begin{align*}
    \SW(\bnz) & = \sum_{i \in [N]} c_i(\bnz) \ge \sum_{i \in \nn_2} c_i(\bnz)
    \ge \beta \cdot \sum_{i \in \nn_2} E_i \cdot q_i \\
    & \ge \beta \cdot \left(\sum_{i \in [N]} q_i \cdot E_i \right) \cdot \frac{\sum_{i \in \nn_2} q_i}{\sum_{i \in [N]} q_i} \tag{\cref{lem:N2}}\\
    & \ge \beta \cdot \left(\sum_{i \in [N]} q_i \cdot E_i + k \cdot E_\cut \right) \cdot \frac{\sum_{i \in [N]} q_i}{k + \sum_{i\in [N]} q_i}\cdot \frac{\sum_{i \in \nn_2} q_i}{\sum_{i \in [N]} q_i} \tag{Since $E_i \ge E_\cut$ for any agent} \\
    & \ge \beta \cdot \SW' \cdot \frac{k}{1 + 1} \cdot \frac{1 - 1 / \alpha}{k} = \SW' \cdot \frac{\beta \cdot (1 - 1 / \alpha)}{2}. 
    \end{align*}
    This completes the proof.
\end{proof}

Since $\SW^{\max} \le \SW'$, it follows that the Price of Anarchy is bounded by
$\frac{2}{\beta \cdot (1 - 1 / \alpha)}$ in this case. 

\subsection{Case 2: Some Agent in \texorpdfstring{$\nn_2$}{N-2} Is a Small Contributor}
Now assume that there exists an agent $i^* \in \nn_2$ for which \[
c_{i^*} < E_{i^*} \cdot \beta \cdot q_{i^*}.
\]
In other words, the agent $i^*$ is not selected with enough value from her equilibrium strategy. Our goal is to show that in this case, the agent $i^*$ can deviate in a way that will improve her payoff if the social welfare is too low, which contradicts the equilibrium condition. 

Fix some constant $\phi > 1$ and consider the top $\phi_{i^*} = \min\{\phi \cdot q_{i^*},\ 1\}$ quantile of $i^*$'s value distribution. Let $\bm{\Pi}_{-i^*} = \bigotimes_{j \neq i^*} \Pi_j$ denote the distribution of signals from all agents except $i^*$. For simplicity, for any signal profile $\bsi$, we write $\rho_{i^*}(\bsi)$ as $\rho_{i^*}(\sigma_{i^*}, \bsi_{-i^*})$ where $\bsi_{-i^*}$ is the signal profile of all agents except $i^*$. For each signal $\sigma_{i^*} \in \Sigma_{i^*}$, define
\[
w_{i^*}(\sigma_{i^*}) = \int_{\bigotimes_{j \neq i^*} \Sigma_j}  \rho_{i^*}(\sigma_{i^*}, \bsi_{-i^*}) \d \bm{\Pi}_{-i^*}(\bsi_{-i^*}).
\]
This represents the overall probability that agent $i^*$ is selected when she sends the signal $\sigma_{i^*}$.

\subsubsection{Construction of the deviation signal \texorpdfstring{$s^*$}{s*}}
This step is the key to the analysis. Choose a constant $\tau \in (0, 1)$. We partition the signals of $\Sigma_{i^*}$ into two sets: 
\[
S_1 = \{\sigma \in \Sigma_{i^*} : w_{i^*}(\sigma) \ge \tau\}, \text{ and } S_2 = \{\sigma \in \Sigma_{i^*} : w_{i^*}(\sigma) < \tau\}. 
\]
Here, $S_1$ consists of signals that yield a high probability (at least $\tau$) of agent $i^*$ being selected, while $S_2$ contains signals with low probabilities of being selected.\footnote{Since combining signals with the same posterior mean does not affect agents' utilities or the social welfare, we can combine all signals with the same posterior mean and index agent $i^*$'s signals by their posterior means. Consequently, we can assume that the signal space for agent $i^*$ is $\Sigma_{i^*} = [0, 1]$ and that $e_{i^*}(\sigma_{i^*}) = \sigma_{i^*}$ for all $\sigma_{i^*} \in [0, 1]$. Moreover, since $w_{i^*}(\sigma)$ is a monotonically increasing function of $\sigma$, both sets $S_1$ and $S_2$ are measurable. This guarantees the valid construction of $\pi_{i^*}'(S \mid v)$ below.} 

We construct a new signaling scheme $\nz_{i^*}' = \{\Sigma_{i^*} \cup \{s^*\}, \pi_{i^*}'\}$ of agent $i^*$. The signaling scheme adds a new signal $s^*$ to the original signal space $\Sigma_{i^*}$. The new signal $s^*$ is formed by merging all the value mass in the top $\phi_{i^*}$ quantile that were originally mapped to $S_2$. Formally, for any measurable set $S \subseteq \Sigma_{i^*}'$ and for each $v \in V$, $\pi_{i^*}(S \mid v)$ is given by:\footnote{For simplicity, we omit the case when the $\phi_{i^*}$-quantile splits a mass point of $\nf_{i^*}$. In that case, we can split the mass point and treat the two parts of masses as two different values (the higher quantile part may map to $s^*$ while the lower quantile part will not map to $s^*$).}
\[
\pi_{i^*}'(S \mid v) = \begin{cases}
    \pi_{i^*}(S \cap S_1 \mid v) + \pi_{i^*}(S_2 \mid v) \cdot \mathbb{1} [s^* \in S] & \text{if $v$ is in the top $\phi_{i^*}$-quantile of $\nf_{i^*}$,}\\
    \pi_{i^*}(S \mid v) & \text{otherwise.}
\end{cases}
\]

The intuition of our construction is that we are ``recombining'' all the signals in $S_2$ that come from the high-value region into one single signal $s^*$. This new signal is intended to improve agent $i^*$'s chances of being selected on higher values by eliminating the low probabilities of being selected that were associated with signals in $S_2$. 

\begin{figure}[tbp]
    \centering
    \includegraphics[width=0.9\linewidth]{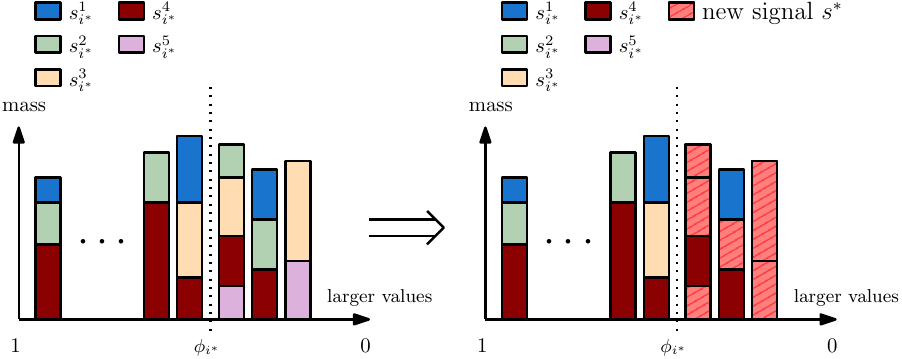}
    \caption{\em Illustration of construction of deviation signaling scheme with signal $s^*$.}
    \label{fig:recombination}
\end{figure}

\cref{fig:recombination} illustrates the process of the construction of $s^*$. The $x$-axis represents the values while the $y$-axis (heights of the bars) represents the mass on the corresponding values. Since each signal consists of masses on different values, we use the same color set of rectangles to represent a signal.  In the illustration, we assume $\Sigma_{i^*}$ consists of five signals $\{s_{i^*}^j\}_{j = 1}^5$. Assume that if signal $s_{i^*}^1$ or $s_{i^*}^4$ is sent, agent $i^*$ has at least $\tau$ probability of being selected. Then $s_{i^*}^1, s_{i^*}^4 \in S_1$. On the other hand, $s_{i^*}^2$, $s_{i^*}^3$, or $s_{i^*}^5$ belong to $S_2$. We will pick all the masses from these signals with values in the top $\phi_{i^*}$ quantile, and combine them into a signal $s^*$, as shown in the tiled area on the right.

\subsubsection{The posterior mean of \texorpdfstring{$s^*$}{s*}}
Let $e_{i^*}(\sigma)$ denote the posterior mean of signal $\sigma$ under agent $i^*$'s original signaling scheme $\nz_{i^*}$, and let $e_{i^*}'(s^*)$ denote the posterior mean of the new signal $s^*$ under $\nz_{i^*}'$. We will first show that the posterior mean of $s^*$ is at least a constant fraction of $E_{i^*}$. 

\begin{lemma}
\label{lem:ei}
    $e_{i^*}'(s^*) \ge E_{i^*} \cdot \left( \frac{1}{\phi} - \frac{\beta}{\tau \cdot \phi} \right)$.
\end{lemma}

\begin{proof}
Since $c_{i^*} < E_{i^*}\cdot \beta \cdot q_{i^*}$ (by the initial assumption of this case) and the value contribution comes from signals in $S_1$ or signals in $S_2$, by considering only signals in $S_1$, we have 
\[
\int_{\sigma \in S_1} e_{i^*}(\sigma) \cdot w_{i^*}(\sigma) \d\, \Pi_{i^*}(\sigma) < E_{i^*}\cdot \beta \cdot q_{i^*}.
\]
By definition of $S_1$, for all $\sigma \in S_1$, $w_{i^*}(\sigma) \ge \tau$. We have 
\begin{equation}
    \int_{\sigma \in S_1} e_{i^*}(\sigma) \d\, \Pi_{i^*}(\sigma) < \frac{E_{i^*}\cdot \beta \cdot q_{i^*}}{\tau}. \label{eqn:contribution_ub}
\end{equation}

Since $s^*$ and the signals in $S_1$ must have ``used up'' all the masses in the top $\phi_{i^*}$ quantile of $\nf_{i^*}$, we have 
\begin{equation}
    \int_{\sigma \in S_1} e_{i^*}(\sigma) \d\, \Pi_{i^*}(\sigma) + e_{i^*}'(s^*) \cdot \Pr[\textrm{$s^*$ is sent}] \ge E_{i^*} \cdot q_{i^*}.  \label{eqn:mass_lb}  
\end{equation}

Since $s^*$ is sent only when agent $i^*$'s value is in the top $\phi_{i^*}$ quantile, it holds that $0 < \Pr[\textrm{$s^*$ is sent}] \le \phi \cdot q_{i^*}$. Combining  with \cref{eqn:contribution_ub,eqn:mass_lb}, we have 
\[
e_{i^*}'(s^*) \ge \frac{E_{i^*} \cdot q_{i^*} - \frac{E_{i^*} \cdot \beta \cdot q_{i^*}}{\tau}}{\phi \cdot q_{i^*}} \ge E_{i^*}\cdot \left(\frac{1}{\phi} - \frac{\beta}{\tau \cdot \phi} \right),
\]
completing the proof.   
\end{proof}

\subsubsection{Deviation to \texorpdfstring{$s^*$}{s*} in new signaling scheme} \label{lem:deviation_wp_lb}
Denote by $\bnz'$ the signaling scheme profile where only agent $i^*$ deviates to the new signaling scheme $\nz_{i^*}'$ while all other agents keep their signaling scheme in $\bnz$. 
\begin{lemma} \label{lem:deviate}
Assuming $\phi \cdot \tau > \alpha,$  in the new signaling scheme profile $\bnz'$, we have 
\[
\Pr[i^* \mbox { is selected} \mid s^* \mbox{ is sent}] < \tau.
\] 
\end{lemma}
\begin{proof}
Assume instead that when signal $s^*$ is sent, agent $i^*$ is selected with probability at least $\tau$. Since $\phi_{i^*} = \min(1, \phi \cdot  q_{i^*})$, we consider the following two cases:

\begin{itemize}
    \item [1.] If $\phi_{i^*} = 1$, then agent $i^*$'s deviation to $s^*$ strictly increases the probability of being selected on all values used to form $s^*$, since their original probability of being selected is less than $\tau$ by definition of $S_2$. Therefore, agent $i^*$ will strictly benefit from the deviation. 
    \item [2.]  If $\phi_{i^*} = \phi \cdot q_{i^*}$, consider the case after agent $i^*$ deviates to $\nz_{i^*}'$. All values in the top $(\phi \cdot q_{i^*})$ quantile are mapped to either signals in $S_1$ or $s^*$. Since when any signal in $S_1$ or $s^*$ is sent to the \principal, agent $i^*$ is selected with probability at least $\tau$, the deviation to $s^*$ makes her overall probability of being selected at least $\phi \cdot q_{i^*} \cdot \tau$. By the construction of $\nz_{i^*}'$, agent $i^*$ loses utility from the cases when her value is at most $v_{i^*}(1 - \phi_{i^*})$ and gains utility from the cases when her value is at least $v_{i^*}(1 - \phi_{i^*})$. 
    Recall that $i^* \in \nn_2$, so $r_{i^*}(\bnz) < \alpha \cdot q_{i^*}$. Since the parameters satisfy $\phi \cdot \tau > \alpha$, agent $i^*$'s overall probability of being selected increases from strictly less than $\alpha \cdot q_{i^*}$ to $\phi \cdot \tau \cdot q_{i^*}$. She will gain at least 
   \[
   (\phi \cdot \tau - \alpha) \cdot q_{i^*} \cdot u_{i^*}(v_{i^*}(1 - \phi_{i^*})) > 0
   \]
   additional utility from the deviation. 
\end{itemize}

Therefore, agent $i^*$ will strictly benefit from the potential strategy deviation from $\nz_{i^*}$ to $\nz_{i^*}'$, which contradicts the fact that $\bnz$ is a Nash equilibrium.
\end{proof}

By the above lemma, we can assume that in $\bnz'$, agent $i^*$ is selected with probability strictly less than $\tau$ when $s^*$ is sent. We now classify all agents into two groups $\nn_\ell$ and $\nn_s$, according to whether an agent $i$ has $E_i$ larger or smaller than $E_{i^*}$:
\[
\nn_\ell = \{i \in [N] : E_i > E_{i^*}\}; \qquad \nn_s = \{i \in [N] : E_i \le E_{i^*}\}.
\]
Since $E_{i^*} \ge E_\cut$, all agents in $\nn_\ell$ must have their $q_i = 1$. Note that $i^* \in \nn_s$. Since $q_{i^*} > 0$ and $\sum_{i \in [N]} q_i = k$, we have $|\nn_\ell| < k$.

Let $M(\bnz)$ denote the random variable indicating the $(k - |\nn_\ell|)$-th largest posterior mean among all agents in $\nn_s$ in the Nash equilibrium $\bnz$. 

\begin{lemma} \label{lem:klargest_win_lb}
If agent $i^*$ is selected with probability less than $\tau$ when $s^*$ is sent in $\bnz'$, we have \[\Pr[M(\bnz) \ge e'_{i^*}(s^*)] \ge 1 - \tau.\] 
\end{lemma}
\begin{proof}
    If not, we have $\Pr[M(\bnz) \ge e'_{i^*}(s^*)] < 1 - \tau$ and thus $\Pr[M(\bnz) < e'_{i^*}(s^*)] \ge \tau$. Denote by $M_{-i^*}(\bnz)$ the random variable representing the $(k - |\nn_\ell|)$-th largest posterior mean among all agents except $i^*$ in $\nn_s$ in $\bnz$. $M_{-i^*}(\bnz)$ is stochastically dominated by $M(\bnz)$. When $M_{-i^*}(\bnz) < e'_{i^*}(s^*)$, agent $i^*$ must be selected when $s^*$ is sent. This contradicts the assumption that she is selected with probability less than $\tau$ when $s^*$ is selected. 
\end{proof}

We first present a proof showing that when the top $(1-\tau)$ quantile of $M(\bnz)$ is no less than the posterior mean of $s^*$, the social welfare decreases by at most a constant. 

\begin{lemma} \label{lem:PoA_M_ub}
     Given a strategy profile $\bnz$, if 
     $\Pr\left[M(\bnz) \ge e'_{i^*}(s^*)\right] \ge 1 - \tau$,
     then we have 
     \[
     \SW(\bnz) \ge \SW' \cdot \frac{(1 - \tau) \cdot \left(\frac{1}{\phi} - \frac{\beta}{\tau \cdot \phi} \right)}{2}.
     \]
\end{lemma}
\begin{proof}
    Consider a hypothetical scenario where, instead of selecting the $k$ agents with the highest posterior means, the \principal selects all agents in $\nn_\ell$ and selects the remaining $k - |\nn_\ell|$ agents from $\nn_s$ with the highest posterior means. Denote the resulting social welfare in this scenario by $\SW^{\textrm{(fake)}}$. We have
    \begin{align}
        \SW^{\textrm{(fake)}}(\bnz) &\ge  (k - |\nn_\ell| ) \cdot \E[M(\bnz)] + \sum_{i \in \nn_\ell} E_{i} \notag\\
        & \geq (1-\tau) \cdot (k - |\nn_\ell| ) \cdot  e'_{i^*}(s^*) +  \sum_{i \in \nn_\ell} E_{i}  \tag{w.p. $1-\tau$, $M(\bnz) \ge e'_{i^*}(s^*)$. }\\
        & \ge (1 - \tau) \cdot (k - |\nn_\ell|) \cdot \left(\frac{1}{\phi} - \frac{\beta}{\tau \cdot \phi} \right) \cdot E_{i^*} + \sum_{i \in \nn_\ell} E_{i}  \tag{\cref{lem:ei}}\\
        & \ge \frac{(1 - \tau) \cdot (k - |\nn_\ell|) \cdot \left(\frac{1}{\phi} - \frac{\beta}{\tau \cdot \phi} \right) \cdot E_{i^*} + \sum_{i \in \nn_\ell} E_{i} }{(1 - \tau) \cdot \left(\frac{1}{\phi} - \frac{\beta}{\tau \cdot \phi} \right) \cdot (k - |\nn_\ell|) + |\nn_\ell|}\cdot (1 - \tau) \cdot \left(\frac{1}{\phi} - \frac{\beta}{\tau \cdot \phi} \right) \cdot k. \label{eqn:fake_lb_1}
    \end{align}
   To see the last inequality, note that  $\alpha > 1$. Since $\tau < 1$, we have $\phi \ge \frac{\alpha}{\tau} > 1$. Since $\beta \le \tau$, we finally have 
   $$(1 - \tau) \cdot \left(\frac{1}{\phi} - \frac{\beta}{\tau \cdot \phi} \right)  \le (1 - \tau) < 1.$$  
Hence, \[(1 - \tau) \cdot \left(\frac{1}{\phi} - \frac{\beta}{\tau \cdot \phi} \right) \cdot (k - |\nn_\ell|) + |\nn_\ell| \ge (1 - \tau) \cdot \left(\frac{1}{\phi} - \frac{\beta}{\tau \cdot \phi} \right) \cdot k.
    \]
    
    Since $E_\cut \le E_i \le E_{i^*} \le E_j$, for any $i \in \nn_s$, and $j \in \nn_\ell$, it follows that any convex combination of $E_{\cut}$ and $\{E_{i}\}_{i \in \nn_s}$ cannot exceed $E_{i^*}$, while $E_{i^*}$ cannot exceed any convex combination of $\{E_{j}\}_{j \in \nn_\ell}$. Therefore, 
    \[
    \frac{k \cdot E_\cut + \sum_{i \in \nn_s} E_i \cdot q_i}{k + \sum_{i \in \nn_s} q_i} \le E_{i^*} \le \frac{\sum_{i \in \nn_{\ell}} E_i}{|\nn_\ell|}.
    \]

    For any real numbers $b \ge a_\ell \ge a_s \ge 0, w_\ell \ge w_s, w_b \ge 0$, we have 
    \[
        \frac{a_\ell \cdot w_s + b \cdot w_b}{w_s + w_b} \ge \frac{a_s \cdot w_\ell + b \cdot w_b}{w_\ell + w_b}.
    \]
    By letting $b = \frac{\sum_{i \in \nn_{\ell}} E_i}{|\nn_\ell|}$, $w_b = |\nn_\ell|$, $a_\ell = E_{i^*}$, $w_s = (1 - \tau) \cdot \left(\frac{1}{\phi} - \frac{\beta}{\tau \cdot \phi} \right) \cdot (k - |\nn_\ell|)$, $a_s = \frac{k \cdot E_\cut + \sum_{i \in \nn_s} E_i \cdot q_i}{k + \sum_{i \in \nn_s} q_i}$, $w_\ell = k + \sum_{i \in \nn_s} q_i$, it follows from \cref{eqn:fake_lb_1} that
    \begin{align*}
        \SW^{\textrm{(fake)}}(\bnz) & \ge \frac{k \cdot E_\cut + \sum_{i \in \nn_s} E_i \cdot q_i + \sum_{i \in \nn_\ell} E_i}{k + \sum_{i \in \nn_s} q_i + |\nn_\ell|} \cdot (1 - \tau) \cdot \left(\frac{1}{\phi} - \frac{\beta}{\tau \cdot \phi} \right) \cdot k \\
        & = \frac{\SW'}{2 k} \cdot (1 - \tau) \cdot \left(\frac{1}{\phi} - \frac{\beta}{\tau \cdot \phi} \right) \cdot k. \tag{Since $q_i = 1$ for any $i \in \nn_\ell$}
    \end{align*}
    
    Since the \principal may not be a posterior mean maximizer, we have $\SW^{\textrm{(fake)}}(\bnz) \le \SW(\bnz)$. This completes the proof. 
\end{proof} 

Since $\SW'$ is an upper bound on social welfare in the full-revelation scenario, and since selection in the hypothetical scenario leads to at most the social welfare of selecting the $k$ agents with highest posterior mean, the PoA in this case is at most $\frac{2}{(1 - \tau) \cdot \left(\frac{1}{\phi} - \frac{\beta}{\tau \cdot \phi} \right)}$.

\begin{proof} [Completing the proof of \cref{thm:main}]
Since either Case 1 or Case 2 happens, the PoA is at most 
\begin{gather*}
    \mbox{PoA} \le \max \left\{\frac{2}{\beta \cdot (1 - 1 / \alpha)} , \frac{2}{(1 - \tau) \cdot \left(\frac{1}{\phi} - \frac{\beta}{\tau \cdot \phi} \right)}\right\}.
\end{gather*}
By setting\footnote{The current parameters are set with $\phi \cdot \tau = \alpha$, we can jitter the parameters a bit to achieve $\phi \cdot \tau > \alpha$ as required. So the PoA bound can be arbitrarily close to $11+5\sqrt{5}$.}
\begin{gather*}
    \alpha = \frac{\sqrt{5} + 1}{2} \approx 1.618, \qquad \beta = \sqrt{5}-2 \approx 0.236, \\
    \tau = \frac{\sqrt{5} - 1}{2} \approx 0.618, \qquad \phi = \frac{\sqrt{5} + 3}{2} \approx 2.618,
\end{gather*}
we know that the PoA of any instance is at most $11 + 5\sqrt{5} \approx 22.181$. This completes the proof of \cref{thm:main}.
\end{proof}

\section{Open Questions}
There are several avenues for future research. The immediate question is to improve the constant factor in \cref{thm:main}; this will require new ideas. %
Further, it would be interesting to extend our results to the setting where the principal solves an arbitrary combinatorial optimization problem on the agents, and when agent values can be multi-dimensional. As an example, each agent can control a subset of edges in a graph, whose cost they seek to signal to the principal, while the principal seeks to compute a minimum spanning tree on these edges. Another example is the price discrimination setting in the work of~\cite*{bergemann2015limits}. Finally, it would be interesting to explore PoA for other models of information revelation beyond signaling, for instance, models where the agents themselves are uncertain about their values~\citep*{RoeslerS}.

\newpage
\bibliographystyle{plainnat}
\bibliography{refs}

\begin{thebibliography}{47}
\providecommand{\natexlab}[1]{#1}
\providecommand{\url}[1]{\texttt{#1}}
\expandafter\ifx\csname urlstyle\endcsname\relax
  \providecommand{\doi}[1]{doi: #1}\else
  \providecommand{\doi}{doi: \begingroup \urlstyle{rm}\Url}\fi

\bibitem[Alijani et~al.(2022)Alijani, Banerjee, Munagala, and
  Wang]{AlijaniBMW22}
Reza Alijani, Siddhartha Banerjee, Kamesh Munagala, and Kangning Wang.
\newblock The limits of an information intermediary in auction design.
\newblock In \emph{Proceedings of the 23rd {ACM} Conference on Economics and
  Computation (EC)}, pages 849--868. {ACM}, 2022.
\newblock \doi{10.1145/3490486.3538370}.
\newblock URL \url{https://doi.org/10.1145/3490486.3538370}.

\bibitem[Anshelevich et~al.(2008)Anshelevich, Dasgupta, Kleinberg, Tardos,
  Wexler, and Roughgarden]{PoS}
Elliot Anshelevich, Anirban Dasgupta, Jon~M. Kleinberg, {\'{E}}va Tardos, Tom
  Wexler, and Tim Roughgarden.
\newblock The price of stability for network design with fair cost allocation.
\newblock \emph{{SIAM} Journal on Computing}, 38\penalty0 (4):\penalty0
  1602--1623, 2008.
\newblock \doi{10.1137/070680096}.
\newblock URL \url{https://doi.org/10.1137/070680096}.

\bibitem[Au and Kawai(2019)]{au2021competitive}
Pak~Hung Au and Keiichi Kawai.
\newblock Competitive disclosure of correlated information.
\newblock \emph{Economic Theory}, 72\penalty0 (3):\penalty0 767–799, January
  2019.
\newblock ISSN 1432-0479.
\newblock \doi{10.1007/s00199-018-01171-7}.
\newblock URL \url{http://dx.doi.org/10.1007/s00199-018-01171-7}.

\bibitem[Au and Kawai(2020)]{AuKawai}
Pak~Hung Au and Keiichi Kawai.
\newblock Competitive information disclosure by multiple senders.
\newblock \emph{Games and Economic Behavior}, 119:\penalty0 56--78, 2020.
\newblock \doi{10.1016/J.GEB.2019.10.002}.
\newblock URL \url{https://doi.org/10.1016/j.geb.2019.10.002}.

\bibitem[Babaioff et~al.(2020)Babaioff, Immorlica, Lucier, and
  Weinberg]{BabaioffILW}
Moshe Babaioff, Nicole Immorlica, Brendan Lucier, and S.~Matthew Weinberg.
\newblock A simple and approximately optimal mechanism for an additive buyer.
\newblock \emph{Journal of the {ACM}}, 67\penalty0 (4):\penalty0 24:1--24:40,
  2020.
\newblock \doi{10.1145/3398745}.
\newblock URL \url{https://doi.org/10.1145/3398745}.

\bibitem[Babichenko et~al.(2021)Babichenko, Talgam{-}Cohen, Xu, and
  Zabarnyi]{DBLP:conf/sigecom/BabichenkoTXZ21}
Yakov Babichenko, Inbal Talgam{-}Cohen, Haifeng Xu, and Konstantin Zabarnyi.
\newblock Regret-minimizing bayesian persuasion.
\newblock In \emph{Proceedings of the 22nd {ACM} Conference on Economics and
  Computation (EC)}, page 128. {ACM}, 2021.
\newblock \doi{10.1145/3465456.3467574}.
\newblock URL \url{https://doi.org/10.1145/3465456.3467574}.

\bibitem[Banerjee et~al.(2021)Banerjee, Kempe, and
  Kleinberg]{banerjee2021threshold}
Siddhartha Banerjee, David Kempe, and Robert Kleinberg.
\newblock Threshold tests as quality signals: Optimal strategies, equilibria,
  and price of anarchy.
\newblock In \emph{Proceedings of the 17th International Conference on Web and
  Internet Economics (WINE)}, pages 299--316. Springer, 2021.
\newblock \doi{10.1007/978-3-030-94676-0_17}.
\newblock URL \url{https://doi.org/10.1007/978-3-030-94676-0_17}.

\bibitem[Banerjee et~al.(2024)Banerjee, Munagala, Shen, and
  Wang]{Banerjee2024fair}
Siddhartha Banerjee, Kamesh Munagala, Yiheng Shen, and Kangning Wang.
\newblock Fair price discrimination.
\newblock In \emph{Proceedings of the 2024 {ACM-SIAM} Symposium on Discrete
  Algorithms (SODA)}, pages 2679--2703. {SIAM}, 2024.
\newblock \doi{10.1137/1.9781611977912.96}.
\newblock URL \url{https://doi.org/10.1137/1.9781611977912.96}.

\bibitem[Banerjee et~al.(2025)Banerjee, Munagala, Shen, and
  Wang]{BanerjeeM0025}
Siddhartha Banerjee, Kamesh Munagala, Yiheng Shen, and Kangning Wang.
\newblock Majorized bayesian persuasion and fair selection.
\newblock In \emph{Proceedings of the 2025 Annual {ACM-SIAM} Symposium on
  Discrete Algorithms (SODA)}, pages 1837--1856. {SIAM}, 2025.
\newblock \doi{10.1137/1.9781611978322.57}.
\newblock URL \url{https://doi.org/10.1137/1.9781611978322.57}.

\bibitem[Bergemann and Morris(2019)]{bergemann2019information}
Dirk Bergemann and Stephen Morris.
\newblock Information design: A unified perspective.
\newblock \emph{Journal of Economic Literature}, 57\penalty0 (1):\penalty0
  44--95, 2019.

\bibitem[Bergemann et~al.(2015)Bergemann, Brooks, and
  Morris]{bergemann2015limits}
Dirk Bergemann, Benjamin Brooks, and Stephen Morris.
\newblock The limits of price discrimination.
\newblock \emph{American Economic Review}, 105\penalty0 (3):\penalty0
  921–957, March 2015.
\newblock ISSN 0002-8282.
\newblock \doi{10.1257/aer.20130848}.
\newblock URL \url{http://dx.doi.org/10.1257/aer.20130848}.

\bibitem[Billingsley(1999)]{billingsley}
Patrick Billingsley.
\newblock \emph{Convergence of probability measures}.
\newblock Wiley Series in Probability and Statistics: Probability and
  Statistics. John Wiley \& Sons Inc., New York, second edition, 1999.
\newblock ISBN 0-471-19745-9.

\bibitem[Boleslavsky and Cotton(2015)]{boleslavsky2015grading}
Raphael Boleslavsky and Christopher Cotton.
\newblock Grading standards and education quality.
\newblock \emph{American Economic Journal: Microeconomics}, 7\penalty0
  (2):\penalty0 248–279, May 2015.
\newblock ISSN 1945-7685.
\newblock \doi{10.1257/mic.20130080}.
\newblock URL \url{http://dx.doi.org/10.1257/mic.20130080}.

\bibitem[Boleslavsky and Cotton(2016)]{boleslavsky2018limited}
Raphael Boleslavsky and Christopher Cotton.
\newblock Limited capacity in project selection: competition through evidence
  production.
\newblock \emph{Economic Theory}, 65\penalty0 (2):\penalty0 385–421, December
  2016.
\newblock ISSN 1432-0479.
\newblock \doi{10.1007/s00199-016-1021-0}.
\newblock URL \url{http://dx.doi.org/10.1007/s00199-016-1021-0}.

\bibitem[Collina et~al.(2025)Collina, Goel, Roth, Ryu, and
  Shi]{collina2025emergentalignmentcompetition}
Natalie Collina, Surbhi Goel, Aaron Roth, Emily Ryu, and Mirah Shi.
\newblock Emergent alignment via competition.
\newblock \emph{CoRR}, abs/2509.15090, 2025.
\newblock \doi{10.48550/ARXIV.2509.15090}.
\newblock URL \url{https://doi.org/10.48550/arXiv.2509.15090}.

\bibitem[Cummings et~al.(2020)Cummings, Devanur, Huang, and
  Wang]{cummings2020algorithmic}
Rachel Cummings, Nikhil~R. Devanur, Zhiyi Huang, and Xiangning Wang.
\newblock Algorithmic price discrimination.
\newblock In \emph{Proceedings of the 2020 {ACM-SIAM} Symposium on Discrete
  Algorithms (SODA)}, pages 2432--2451. {SIAM}, 2020.
\newblock \doi{10.1137/1.9781611975994.149}.
\newblock URL \url{https://doi.org/10.1137/1.9781611975994.149}.

\bibitem[Devic et~al.(2024)Devic, Korolova, Kempe, and
  Sharan]{devic2024stability}
Siddartha Devic, Aleksandra Korolova, David Kempe, and Vatsal Sharan.
\newblock Stability and multigroup fairness in ranking with uncertain
  predictions.
\newblock In \emph{Proceedings of the 41st International Conference on Machine
  Learning (ICML)}, volume 235 of \emph{Proceedings of Machine Learning
  Research}, pages 10661--10686. {PMLR} / OpenReview.net, 2024.
\newblock URL \url{https://proceedings.mlr.press/v235/devic24a.html}.

\bibitem[Ding et~al.(2023)Ding, Feng, Ho, Tang, and Xu]{ding2023competitive}
Bolin Ding, Yiding Feng, Chien{-}Ju Ho, Wei Tang, and Haifeng Xu.
\newblock Competitive information design for pandora's box.
\newblock In \emph{Proceedings of the 2023 {ACM-SIAM} Symposium on Discrete
  Algorithms (SODA)}, pages 353--381. {SIAM}, 2023.
\newblock \doi{10.1137/1.9781611977554.CH15}.
\newblock URL \url{https://doi.org/10.1137/1.9781611977554.ch15}.

\bibitem[Dixon-Rom\'{a}n et~al.(2013)Dixon-Rom\'{a}n, Everson, and
  McArdle]{dixon2013race}
Ezekiel~J. Dixon-Rom\'{a}n, Howard~T. Everson, and John~J. McArdle.
\newblock Race, poverty and sat scores: Modeling the influences of family
  income on black and white high school students’ sat performance.
\newblock \emph{Teachers College Record: The Voice of Scholarship in
  Education}, 115\penalty0 (4):\penalty0 1–33, April 2013.
\newblock ISSN 1467-9620.
\newblock \doi{10.1177/016146811311500406}.
\newblock URL \url{http://dx.doi.org/10.1177/016146811311500406}.

\bibitem[Du et~al.(2024)Du, Tang, Wang, and Zhang]{Du2024}
Zhicheng Du, Wei Tang, Zihe Wang, and Shuo Zhang.
\newblock Competitive information design with asymmetric senders.
\newblock In \emph{Proceedings of the 25th {ACM} Conference on Economics and
  Computation (EC)}, page 1203. {ACM}, 2024.
\newblock \doi{10.1145/3670865.3673474}.
\newblock URL \url{https://doi.org/10.1145/3670865.3673474}.

\bibitem[Dughmi(2017)]{dughmi2017algorithmic}
Shaddin Dughmi.
\newblock Algorithmic information structure design: a survey.
\newblock \emph{SIGecom Exchanges}, 15\penalty0 (2):\penalty0 2--24, 2017.
\newblock \doi{10.1145/3055589.3055591}.
\newblock URL \url{https://doi.org/10.1145/3055589.3055591}.

\bibitem[Dughmi and Xu(2016)]{dughmi2016algorithmic}
Shaddin Dughmi and Haifeng Xu.
\newblock Algorithmic bayesian persuasion.
\newblock In \emph{Proceedings of the 48th Annual {ACM} {SIGACT} Symposium on
  Theory of Computing (STOC)}, pages 412--425. {ACM}, 2016.
\newblock \doi{10.1145/2897518.2897583}.
\newblock URL \url{https://doi.org/10.1145/2897518.2897583}.

\bibitem[Emelianov et~al.(2020)Emelianov, Gast, Gummadi, and
  Loiseau]{emelianov2020fair}
Vitalii Emelianov, Nicolas Gast, Krishna~P. Gummadi, and Patrick Loiseau.
\newblock On fair selection in the presence of implicit variance.
\newblock In \emph{Proceedings of the 21st {ACM} Conference on Economics and
  Computation (EC)}, pages 649--675. {ACM}, 2020.
\newblock \doi{10.1145/3391403.3399482}.
\newblock URL \url{https://doi.org/10.1145/3391403.3399482}.

\bibitem[Garc{\'{\i}}a{-}Soriano and Bonchi(2021)]{garciasoriano2021maxmin}
David Garc{\'{\i}}a{-}Soriano and Francesco Bonchi.
\newblock Maxmin-fair ranking: Individual fairness under group-fairness
  constraints.
\newblock In \emph{Proceedings of 27th {ACM} {SIGKDD} Conference on Knowledge
  Discovery and Data Mining (KDD)}, pages 436--446. {ACM}, 2021.
\newblock \doi{10.1145/3447548.3467349}.
\newblock URL \url{https://doi.org/10.1145/3447548.3467349}.

\bibitem[Gentzkow and Kamenica(2016)]{gentzkow2016competition}
Matthew Gentzkow and Emir Kamenica.
\newblock Competition in persuasion.
\newblock \emph{The Review of Economic Studies}, 84\penalty0 (1):\penalty0
  300–322, October 2016.
\newblock ISSN 1467-937X.
\newblock \doi{10.1093/restud/rdw052}.
\newblock URL \url{http://dx.doi.org/10.1093/restud/rdw052}.

\bibitem[Gentzkow and Kamenica(2017)]{gentzkow2017bayesian}
Matthew Gentzkow and Emir Kamenica.
\newblock Bayesian persuasion with multiple senders and rich signal spaces.
\newblock \emph{Games and Economic Behavior}, 104:\penalty0 411--429, 2017.
\newblock \doi{10.1016/J.GEB.2017.05.004}.
\newblock URL \url{https://doi.org/10.1016/j.geb.2017.05.004}.

\bibitem[Glicksberg(1952)]{glicksberg1952further}
I.~L. Glicksberg.
\newblock A further generalization of the kakutani fixed point theorem, with
  application to nash equilibrium points.
\newblock \emph{Proceedings of the American Mathematical Society}, 3\penalty0
  (1):\penalty0 170, February 1952.
\newblock ISSN 0002-9939.
\newblock \doi{10.2307/2032478}.
\newblock URL \url{http://dx.doi.org/10.2307/2032478}.

\bibitem[Gradwohl et~al.(2022)Gradwohl, Hahn, Hoefer, and
  Smorodinsky]{gradwohl2022reaping}
Ronen Gradwohl, Niklas Hahn, Martin Hoefer, and Rann Smorodinsky.
\newblock Reaping the informational surplus in bayesian persuasion.
\newblock \emph{American Economic Journal: Microeconomics}, 14\penalty0
  (4):\penalty0 296–317, November 2022.
\newblock ISSN 1945-7685.
\newblock \doi{10.1257/mic.20200399}.
\newblock URL \url{http://dx.doi.org/10.1257/mic.20200399}.

\bibitem[Hahn et~al.(2020)Hahn, Hoefer, and Smorodinsky]{hahn2020prophet}
Niklas Hahn, Martin Hoefer, and Rann Smorodinsky.
\newblock Prophet inequalities for bayesian persuasion.
\newblock In \emph{Proceedings of the 29th International Joint Conference on
  Artificial Intelligence (IJCAI)}, pages 175--181. ijcai.org, 2020.
\newblock \doi{10.24963/IJCAI.2020/25}.
\newblock URL \url{https://doi.org/10.24963/ijcai.2020/25}.

\bibitem[Hahn et~al.(2022)Hahn, Hoefer, and Smorodinsky]{hahn2022secretary}
Niklas Hahn, Martin Hoefer, and Rann Smorodinsky.
\newblock The secretary recommendation problem.
\newblock \emph{Games and Economic Behavior}, 134:\penalty0 199--228, 2022.
\newblock \doi{10.1016/J.GEB.2022.05.002}.
\newblock URL \url{https://doi.org/10.1016/j.geb.2022.05.002}.

\bibitem[Hartline et~al.(2014)Hartline, Hoy, and Taggart]{HartlineHT14}
Jason~D. Hartline, Darrell Hoy, and Sam Taggart.
\newblock Price of anarchy for auction revenue.
\newblock In \emph{Proceedings of the 15th {ACM} Conference on Economics and
  Computation (EC)}, pages 693--710. {ACM}, 2014.
\newblock \doi{10.1145/2600057.2602878}.
\newblock URL \url{https://doi.org/10.1145/2600057.2602878}.

\bibitem[Hoffmann et~al.(2020)Hoffmann, Inderst, and
  Ottaviani]{hoffmann2020persuasion}
Florian Hoffmann, Roman Inderst, and Marco Ottaviani.
\newblock Persuasion through selective disclosure: Implications for marketing,
  campaigning, and privacy regulation.
\newblock \emph{Management Science}, 66\penalty0 (11):\penalty0 4958--4979,
  2020.
\newblock \doi{10.1287/MNSC.2019.3455}.
\newblock URL \url{https://doi.org/10.1287/mnsc.2019.3455}.

\bibitem[Hossain et~al.(2024)Hossain, Wang, Lin, Chen, Parkes, and
  Xu]{hossain2024multi}
Safwan Hossain, Tonghan Wang, Tao Lin, Yiling Chen, David~C. Parkes, and
  Haifeng Xu.
\newblock Multi-sender persuasion: {A} computational perspective.
\newblock In \emph{Proceedings of the 41th International Conference on Machine
  Learning (ICML)}, volume 235 of \emph{Proceedings of Machine Learning
  Research}, pages 18944--18971. {PMLR} / OpenReview.net, 2024.
\newblock URL \url{https://proceedings.mlr.press/v235/hossain24c.html}.

\bibitem[Hwang et~al.(2023)Hwang, Kim, and Boleslavsky]{hwang2019competitive}
Ilwoo Hwang, Kyungmin Kim, and Raphael Boleslavsky.
\newblock Competitive advertising and pricing.
\newblock Working paper, revise-and-resubmit, 2023.
\newblock URL
  \url{https://static1.squarespace.com/static/5271497de4b03475d0dd1240/t/649d4f14e3663901f6056f2b/1688030997217/Competitive_Advertising_and_Pricing_012023.pdf}.

\bibitem[Jain and Whitmeyer(2019)]{jain2019competing}
Vasudha Jain and Mark Whitmeyer.
\newblock Competing to persuade a rationally inattentive agent.
\newblock \emph{SSRN Electronic Journal}, 2019.
\newblock ISSN 1556-5068.
\newblock \doi{10.2139/ssrn.3424451}.
\newblock URL \url{http://dx.doi.org/10.2139/ssrn.3424451}.

\bibitem[Kamenica and Gentzkow(2011)]{kamenica2011bayesian}
Emir Kamenica and Matthew Gentzkow.
\newblock Bayesian persuasion.
\newblock \emph{American Economic Review}, 101\penalty0 (6):\penalty0
  2590–2615, October 2011.
\newblock ISSN 0002-8282.
\newblock \doi{10.1257/aer.101.6.2590}.
\newblock URL \url{http://dx.doi.org/10.1257/aer.101.6.2590}.

\bibitem[Koessler et~al.(2022)Koessler, Laclau, Renault, and
  Tomala]{koessler2022long}
Frederic Koessler, Marie Laclau, Jérôme Renault, and Tristan Tomala.
\newblock Long information design.
\newblock \emph{Theoretical Economics}, 17\penalty0 (2):\penalty0 883–927,
  2022.
\newblock ISSN 1933-6837.
\newblock \doi{10.3982/te4557}.
\newblock URL \url{http://dx.doi.org/10.3982/te4557}.

\bibitem[Krengel and Sucheston(1978)]{KrengelS}
Ulrich Krengel and Louis Sucheston.
\newblock On semiamarts, amarts, and processes with finite value.
\newblock In \emph{Probability on {B}anach spaces}, volume~4 of \emph{Adv.
  Probab. Related Topics}, pages 197--266. Dekker, New York, 1978.
\newblock ISBN 0-8247-6799-3.

\bibitem[Mehrotra and Celis(2021)]{mehrotra2021mitigating}
Anay Mehrotra and L.~Elisa Celis.
\newblock Mitigating bias in set selection with noisy protected attributes.
\newblock In \emph{Proceedings of the 2021 {ACM} Conference on Fairness,
  Accountability, and Transparency (FAccT)}, pages 237--248. {ACM}, 2021.
\newblock \doi{10.1145/3442188.3445887}.
\newblock URL \url{https://doi.org/10.1145/3442188.3445887}.

\bibitem[Ravindran and Cui(2022)]{ravindran2020competing}
Dilip Ravindran and Zhihan Cui.
\newblock Competing persuaders in zero-sum games.
\newblock \emph{SSRN Electronic Journal}, 2022.
\newblock ISSN 1556-5068.
\newblock \doi{10.2139/ssrn.4241719}.
\newblock URL \url{http://dx.doi.org/10.2139/ssrn.4241719}.

\bibitem[Roesler and Szentes(2017)]{RoeslerS}
Anne-Katrin Roesler and Balázs Szentes.
\newblock Buyer-optimal learning and monopoly pricing.
\newblock \emph{American Economic Review}, 107\penalty0 (7):\penalty0
  2072–2080, July 2017.
\newblock ISSN 0002-8282.
\newblock \doi{10.1257/aer.20160145}.
\newblock URL \url{http://dx.doi.org/10.1257/aer.20160145}.

\bibitem[Roughgarden(2015)]{smoothness}
Tim Roughgarden.
\newblock Intrinsic robustness of the price of anarchy.
\newblock \emph{Journal of the {ACM}}, 62\penalty0 (5):\penalty0 32:1--32:42,
  2015.
\newblock \doi{10.1145/2806883}.
\newblock URL \url{https://doi.org/10.1145/2806883}.

\bibitem[Sapiro-Gheiler(2023)]{sapiro2024persuasion}
Eitan Sapiro-Gheiler.
\newblock Persuasion with ambiguous receiver preferences.
\newblock \emph{Economic Theory}, 77\penalty0 (4):\penalty0 1173–1218,
  September 2023.
\newblock ISSN 1432-0479.
\newblock \doi{10.1007/s00199-023-01522-z}.
\newblock URL \url{http://dx.doi.org/10.1007/s00199-023-01522-z}.

\bibitem[Shen et~al.(2023)Shen, Wang, Zhu, Fain, and
  Munagala]{shen2023fairness}
Zeyu Shen, Zhiyi Wang, Xingyu Zhu, Brandon Fain, and Kamesh Munagala.
\newblock Fairness in the assignment problem with uncertain priorities.
\newblock In \emph{Proceedings of the 2023 International Conference on
  Autonomous Agents and Multiagent Systems (AAMAS)}, pages 188--196. {ACM},
  2023.
\newblock \doi{10.5555/3545946.3598636}.
\newblock URL \url{https://dl.acm.org/doi/10.5555/3545946.3598636}.

\bibitem[Singh et~al.(2021)Singh, Kempe, and Joachims]{singh2021fairness}
Ashudeep Singh, David Kempe, and Thorsten Joachims.
\newblock Fairness in ranking under uncertainty.
\newblock In \emph{Proceedings of the 35th Annual Conference on Neural
  Information Processing Systems (NeurIPS)}, pages 11896--11908, 2021.
\newblock URL
  \url{https://proceedings.neurips.cc/paper/2021/hash/63c3ddcc7b23daa1e42dc41f9a44a873-Abstract.html}.

\bibitem[Tang et~al.(2024)Tang, Xu, Zhang, and Zhu]{tang2024intrinsic}
Wei Tang, Haifeng Xu, Ruimin Zhang, and Derek Zhu.
\newblock Intrinsic robustness of prophet inequality to strategic reward
  signaling.
\newblock In \emph{Annual Conference on Neural Information Processing Systems
  (NeurIPS)}, 2024.
\newblock URL
  \url{http://papers.nips.cc/paper_files/paper/2024/hash/f3e644506dad33613919fa85af6665d0-Abstract-Conference.html}.

\bibitem[Xu et~al.(2015)Xu, Rabinovich, Dughmi, and
  Tambe]{DBLP:conf/aaai/XuRDT15}
Haifeng Xu, Zinovi Rabinovich, Shaddin Dughmi, and Milind Tambe.
\newblock Exploring information asymmetry in two-stage security games.
\newblock In \emph{Proceedings of the 29th {AAAI} Conference on Artificial
  Intelligence (AAAI)}, pages 1057--1063. {AAAI} Press, 2015.
\newblock \doi{10.1609/AAAI.V29I1.9290}.
\newblock URL \url{https://doi.org/10.1609/aaai.v29i1.9290}.

\end{thebibliography}

\newpage
\phantomsection
\addcontentsline{toc}{section}{Appendices}
\appendix
\section{PoA of a Discretized Game} \label{app:discretized}
In this section, we present a natural discretized version of the information disclosure game where the Nash Equilibrium always exists and our PoA bound still holds.

\subsection{A Restricted Signaling Model}

Let there be $N$ agents $\{1, 2, \ldots, N\}$. Let $V = \{v_{(1)}, v_{(2)}, \ldots, v_{(M)}\}$ be a finite set of possible values of all agents. The prior distribution of $i$'s value, denoted by $\nf_i$, is a discrete distribution over these values. Same as the main text, we assume that $\nf_i$'s are independent. Fix a small number $\epsilon > 0$. We assume that in each $\nf_i$, the probability mass on each value is a multiple of $\epsilon$. Formally, $\Pr[v_i = v_{(j)}] = t_{iv} \cdot \epsilon$, where $t_{iv} \in \mathbb{Z}^{\ge 0}$ is a non-negative integer. 

Let $\nz_i = (\Sigma_i, \pi_i)$ be the signaling scheme of agent $i$. In the probabilistic mapping $\pi_i$, we further assume that for each signal $\sigma_i \in \Sigma_i$, the conditional probability $\Pr[\sigma_i \textrm{ is sent} \mid v_i = v_{(j)}]$ is a multiple of $\frac{1}{t_{iv}}$ (if $t_{iv} \ge 1$). In other words, the signaling scheme is also discretized in the sense that each signal $\sigma_i$ is formed by contraction from $\epsilon$-multiple masses on each value. 

Since the number of signaling schemes is finite in this setting, the NE must exist. Note that in the NE, each agent may mix different signaling schemes. We call this the \emph{$\epsilon$-discretized Bayesian persuasion game}, and will show that in this model, the PoA is still a constant.

\subsection{Main Result for Restricted Signaling Spaces}
Similar to \cref{thm:main}, we will prove the following result.
\begin{theorem}
\label{rebuttal:main}
The PoA of the $\epsilon$-discretized Bayesian persuasion game with $k$ selection (for any $k \ge 1$) is at most $22.459 / (1 - 2.3766  \cdot \epsilon N)$, which is close to $22.459$ when $\epsilon \ll 1/N$.
\end{theorem}

The proof idea is almost identical to \cref{sec:general} of our paper. In this section, we specify the places where the proof steps are different for the discretized game. 

\paragraph{Change in the top quantile cuts.} In \cref{sec:quantile_cuts}, we additionally let $\hat{q}_i = \left\lfloor \frac{q_i}{\epsilon} \right\rfloor \cdot \epsilon$ be the ``discretized'' quantile cut, and we let $\hat{E}_i$ be the mean of the top $\hat{q}_i$ quantile for agent $i$. Note that the truncation will only take effect on those agents with $E_i = E_{\cut}$ (otherwise their $q_i = 1$ and thus $q_i = \hat{q}_i$). Moreover, we have $\hat{q}_i \ge q_i - \epsilon$ for each agent, and the truncated part has mean at most $E_{\cut}$.

Now we modify the grouping of agents so that 
\[\nn_1 = \{i \in [N]: r_i(\bnz) \ge \alpha \cdot \hat{q}_i\}, \quad \nn_2 = \{i \in [N]: r_i(\bnz) < \alpha \cdot \hat{q}_i\}.\] \cref{lem:swub} will remain the same, so that $\SW'$ is still upper bounded by $E_{\cut} \cdot k + \sum_{i \in [N]} E_i \cdot q_i$. \cref{lem:large_contributors} will be modified to the following:
\begin{lemma} [Modified version of \cref{lem:large_contributors}]
    Given a strategy profile $\bnz$, if all agents in $\nn_2$ satisfy $c_i(\bnz) \ge \hat{E}_i \cdot \beta \cdot \hat{q}_i$, then we have 
    \[
        \SW(\bnz) \ge \beta \cdot \SW' \cdot \frac{1 - 1 / \alpha - \epsilon \cdot N}{2}.
    \]
\end{lemma}

Since after the truncation process, the expected values in the discretized quantile cuts $\hat{q}_i$'s can be different, we can not directly argue that $\min_{i\in \nn_2} \hat{E}_i \ge \max_{i \in \nn_1} \hat{E}_i$ (as stated in \cref{lem:N2}). However, we can use the original $\sum_{i\in \nn_1}{q_i} \cdot E_{\cut}$ as a bridge to prove the modified version of \cref{lem:large_contributors}.

\begin{proof} [Proof of modified version of \cref{lem:large_contributors}]
    If $E_i > E_{\cut}$ for some agent $i$, both $q_i$ and $q_{i^*}$ will be $1$. Since $\alpha$ is a parameter strictly larger than $1$, $i$ must lie in $\nn_2$. Therefore, for any agent $i'$ in $\nn_1$, we have $E_{i'} \le E_{\cut}$. Moreover, since $\hat{q}_i$ is truncated from $q_i$. We have 
    \[
    \hat{q}_i \cdot \hat{E}_i \le q_i \cdot E_i \le q_i \cdot E_{\cut}.
    \]

    Furthermore, we have
    \begin{align*}
    \SW(\bnz) & = \sum_{i \in [N]} c_i \ge \sum_{i \in \nn_2} c_i
    \ge \beta \cdot \sum_{i \in \nn_2} \hat{E}_i \cdot \hat{q}_i \\
    & \ge \beta \cdot \left( \sum_{i \in \nn_2} E_i \cdot q_i  - \epsilon \cdot |\nn_2| \cdot E_{\cut} \right) \\
    & \ge \beta \cdot \frac{\sum_{i \in [N]} q_i \cdot E_i + k \cdot E_\cut}{\sum_{i \in [N]} q_i + k} \cdot \left(\sum_{i \in \nn_2} q_i - \epsilon \cdot |N_2| \right)\tag{Since for all $i \in \nn_2$, $E_i \ge E_\cut$} \\
    & \ge \beta \cdot \SW' \cdot \frac{1 - 1 / \alpha - \epsilon \cdot N}{2}. \tag{Since $\sum_{i \in \nn_2} q_i \ge (1 - 1/\alpha) \cdot k$}
    \end{align*}
    This completes the proof.
\end{proof}    

\paragraph{Change in the construction of $s^*$.} 
Similar to Section 5.3.1, we consider the case when there is an agent $i^*$ with small contribution to welfare so that $c_{i^*} < \hat{E}_{i^*} \cdot \beta \cdot \hat{q}_{i^*}$. We partition all the signals sent by $i^*$ into two parts $S_1$ and $S_2$, where $S_2$ consists of all signals with a winning probability strictly less than $\tau$. Next, we combine all the masses on the top $\phi \cdot \hat{q}_{i^*}$ (rather than the original $\phi \cdot q_{i^*}$) quantile of $i^*$ that are used in signals in $S_2$ into a single signal $s^*$. %

Additionally, we restrict the constant parameter $\phi$ to be an integer, so that the quantile cut $\phi \cdot \hat{q}_{i^*}$ is still a multiple of $\epsilon$ and the deviation to $s^*$ is valid. \cref{lem:ei} will become: 
$$e_{i^*}'(s^*) \ge \hat{E}_{i^*} \cdot \left(\frac{1}{\phi} - \frac{\beta}{\tau \cdot \phi}\right).$$ 
\cref{lem:deviate,lem:klargest_win_lb,lem:PoA_M_ub} will remain the same. Note that when proving \cref{lem:PoA_M_ub}, we will use the original set $\nn_s$ and $\nn_\ell$. The only change in the proof is an additional usage of $\hat{E}_{i^*} \ge E_{i^*}$. In the first set of inequalities in the proof of \cref{lem:PoA_M_ub}, we add
\begin{align*}
(1-\tau) \cdot (k - |\nn_\ell| ) \cdot  e'_{i^*}(s^*) +  \sum_{i \in \nn_\ell} E_{i} &\ge (1 - \tau) \cdot (k - |\nn_\ell|) \cdot \left(\frac{1}{\phi} - \frac{\beta}{\tau \cdot \phi} \right) \cdot \hat{E}_{i^*} + \sum_{i \in \nn_\ell} E_{i} \tag{Since $e_{i^*}'(s^*) \ge \hat{E}_{i^*} \cdot \left(\frac{1}{\phi} - \frac{\beta}{\tau \cdot \phi}\right)$}\\
&\ge (1 - \tau) \cdot (k - |\nn_\ell|) \cdot \left(\frac{1}{\phi} - \frac{\beta}{\tau \cdot \phi} \right) \cdot E_{i^*} + \sum_{i \in \nn_\ell} E_{i} \tag{Since $\hat{E}_{i^*}\ge E_{i^*}$}
\end{align*}
and the inequality set can go through from the second line to the third line.  

Assume $\epsilon$ is small enough. By setting 
\begin{gather*}
    \alpha = 1.7264, \qquad \beta = 0.21164, \qquad \tau = 0.57883, \qquad \phi = 3,
\end{gather*}
the PoA of such discretized setting is at most $22.459 /(1 - 2.3766  \cdot \epsilon N)$. The slight increment on the number comes from restricting the parameter $\phi$ to be integral. 

\section{PoA of a Noisy Selection Game} \label{app:noisy}
We consider a game with $N$ agents where each agent has a private value $v_i$ drawn from a prior distribution $\nf_i$ on $I = [\underline{v}, 1]$ with $\underline{v} > 0$.\footnote{This lower bound on the value prevents the agent from putting a point mass on value $0$ and ensures the smoothing effect of the multiplicative noise.} Each agent chooses a signaling scheme, which we model it as a probability measure on value-message pairs. The principal observes noisy versions of the agents' posterior means and selects the top $k$ agents. In this section, we prove the existence of a Nash equilibrium using Glicksberg's theorem, and apply our PoA result to this setting.

\subsection{Model and Preliminaries}
\subsubsection{Signaling scheme re-formalization}
Each agent $i$ chooses a \textbf{signaling scheme}. In this section, we formulate the signaling scheme as a distribution of the value-signal pair. A signaling scheme is a probability measure $\pi_i$ on $I_v \times I_m$. Note that $I_v = I_m = I = [\underline{v}, 1]$. $I_v$ represents the agent's value space, and the $I_m$ represents the signal space.\footnote{Here we use the posterior mean to represent a signal (and we combine signals with the same posterior mean).} Each $\pi_i$ must satisfy the following two constraints:
\begin{itemize}
\item[(a)] The marginal distribution on $I_v$ is $\nf_i$. Formally, for every measurable $A \subseteq I_v$,
$$\pi_i(A \times I_m) = \nf_i(A).$$ 
\item[(b)] The conditional expectation satisfies $\E[v_i \mid m_i] = m_i$ almost surely. Formally, for any continuous function $h$ on $I_m$,  
$$\int_{I_v \times I_m} (v - m) \cdot h(m) \d \pi_i(v, m) = 0.$$
\end{itemize}

We define the strategy space for agent $i$ as:
\[S_i = \left\{\pi_i \in \Delta(I_v \times I_m) \mid \textrm{conditions (a) and (b) hold}\right\}.\] 
One can verify that the above form is equivalent to the probabilistic mapping form that we used to define signaling schemes in \cref{sec:prelim_signals}.

\subsubsection{Noisy Selection and Utilities}
After agents choose their signaling schemes, the pairs $(v_i, m_i)$ are drawn independently from $\pi_i$. The principal observes $y_i = m_i \cdot \xi_i$ as the value of agent $i$, where $\xi_i$'s are sampled from the same uniform distribution $\Xi = \mathtt{Uniform}[1 - \eta, 1 + \eta]$ independently. \footnote{Our proof directly extends to heterogenuous atomless distributions, that is, each $\xi_i$ is drawn from $\Xi_i \in \Delta([1 - \eta, 1 + \eta])$ where $\Xi(\{t\}) = 0$ for any $t \in [1 - \eta, 1 + \eta]$, but for simplicity we use uniform distribution here.}

The principal selects the $k$ agents with the highest observed value $y_i$ and breaks ties uniformly. If an agent $i$ is selected when her value is $v_i$, she will receive utility $u_i(v)$, where $u_i$ is continuous and weakly increasing; otherwise, she receives $0$. Let $\bv = (v_1, \ldots, v_N)$, $\bm = (m_1, \ldots, m_N)$ and $\bxi = (\xi_1, \ldots, \xi_N)$ denote the value, message and noise profile of all agents. Furthermore, denote by $\rho_i(\bem, \bxi)$ the probability of $i$ being selected when the revealed messages and noises are $\bem$ and $\bxi$, we have 

$$
\rho_i(\bem, \bxi) = \begin{cases}
    1 & \text{ if $|\{j \in [N] \mid \xi_j \cdot m_j \ge \xi_i \cdot m_i\}| \le k$}; \\
    \frac{\max\{0,\ k - |\{j \in [N] \mid \xi_j \cdot m_j > \xi_i \cdot m_i\}|}{|\{j \in [N] \mid \xi_j \cdot m_j = \xi_i \cdot m_i\}|}& \text{ else}.
\end{cases} 
$$

The utility of an agent after revelation is thus $g_i(\bv, \bem, \bxi) = \rho_i(\bem, \bxi) \cdot u_i(v_i)$. The expected utility pf agent $i$ on the agents' strategy profile $\bmpi = (\pi_1, \ldots, \pi_N)$ is given by
$$
U_i(\bmpi) = \E_{\{(v_i, m_i) \sim \pi_i,\ \xi_i \sim \Xi\}_{i \in [N]}}[g_i(\bv, \bem, \bxi)].
$$

We call the game defined above the \emph{$\eta$-noisy Bayesian persuasion game}. 

\subsection{Existence of Nash Equilibrium}
In the original game, the payoff function is discontinuous, so an equilibrium is not guaranteed. In this section, we will show  that adding continuously-distributed noises to the observed values eliminates the discontinuity of the payoff function, and hence guarantees a Nash equilibrium. The challenge here is that since our utility functions are general, the non-linearity of the utility functions makes the agents' payoffs not only depend on their posterior distributions. So we cannot assume that each agent's strategy space is without loss of generality the mean preserving contraction of their priors as in the work of \citet*{Du2024}. 

\begin{theorem} \label{thm:noisy_nash}
    The Nash equilibrium always exists for the $\eta$-noisy Bayesian persuasion game.
\end{theorem}

In order to prove the Nash equilibrium existence, we need to utilize the Glicksberg's Theorem \citep*{glicksberg1952further}. We first prove that the strategy space is a non-empty, compact and convex metric space.

\begin{lemma} \label{lem:noisy_compact}
    For each agent $i$, the strategy space $S_i$ is non-empty, compact and convex.
\end{lemma}

\begin{proof}
    Since the uninformative signaling scheme where $\pi_i = \nf_i \times \delta_{\E_{v_i \sim \nf_i}[v_i]}$\footnote{$\delta_x$ represents the Dirac delta measure on a single point $x$.} satisfies both conditions, the strategy space is non-empty. 
    
    We endow the $\Delta(I_v \times I_m)$ with the weak topology. Let $\big\{\pi_i^{(t)}\big\}$ weakly converge to $\pi_i$. Since the marginal distribution on $I_v$ is $\nf_i$ (condition (a)), for any continuous $f$ on $I_v$ we have \[\int_{I_v \times I_m} f(v) \d \pi_i^{(t)}(v, m) = \int_{I_v} f(v) \d \nf_i(v).\] By the property of weakly convergence, we have the left hand side converges to $\int_{I_v \times I_m} f(v) \d \pi_i(v, m)$, hence $\pi_i \in S_i$. Therefore, the space $S_i$ is closed. Since $\Delta(I_v \times I_m) = \Delta(I \times I)$ is compact and a closed subset of a compact set is compact, we have $S_i$ is compact.

    Finally, we prove that $S_i$ is convex. Let $\pi_i^{(1)},\pi_i^{(2)}\in S_i$ and fix $\lambda\in[0,1]$. Define $\pi:=\lambda \pi_i^{(1)}+(1-\lambda)\pi_i^{(2)}$.
    For any continuous $f$ on $I_v$, we have
    \begin{align*}
    \int_{I_v\times I_m} f(v) \d \pi(v, m) &= \lambda \cdot \int_{I_v\times I_m} f(v) \d \pi_i^{(1)}(v, m) + (1-\lambda) \cdot \int_{I_v\times I_m} f(v) \d \pi_i^{(2)}(v, m) \\
    &= \lambda \cdot \int_{I_v} f(v) \d \nf_i(v) + (1-\lambda) \cdot \int_{I_v} f(v) \d \nf_i(v) \\
    &= \int_{I_v} f(v) \d \nf_i(v),
    \end{align*}
    so the $I_v$-marginal of $\pi$ is $\nf_i$, $\pi$ satisfies the constraint (a) of $S_i$. 

    Similarly, for any continuous $h$ on $I_m$, we have $\int_{I_v \times I_m} (v - m) \cdot h(m) \d \pi(v, m) = \lambda \cdot \int_{I_v \times I_m} (v - m) \cdot h(m) \d \pi_i^{(1)}(v, m) + (1 - \lambda) \cdot \int_{I_v \times I_m} (v - m) \cdot h(m) \d \pi_i^{(2)}(v, m) = 0 + 0 = 0$. Hence $\pi$ satisfies the posterior mean constraint (b) of $S_i$. Therefore, we have proved that $\pi\in S_i$, and $S_i$ is thus convex.  
\end{proof}

\begin{lemma} \label{lem:noisy_continuous}
    For each agent $i$, the utility $U_i:S_1 \times \cdots \times S_N \rightarrow \R$ is continuous.
\end{lemma}
\begin{proof}
    Let $(\pi_1^{(t)}, \pi_2^{(t)}, \ldots, \pi_N^{(t)})$ converges to $(\pi_1, \pi_2, \ldots, \pi_N)$ weakly. Define the following outcome product measures:
    \begin{gather*}
        \mu^{(t)} = \pi_1^{(t)} \times \pi_2^{(t)} \times \cdots \times \pi_N^{(t)} \times \Xi^{N}; \\
        \mu = \pi_1 \times \pi_2 \times \cdots \times \pi_N \times \Xi^{N}.
    \end{gather*}

    To prove the continuity of $U_i$, it suffices to show that 
    \[
    U_i(\bmpi^{(t)}) = \int g_i(\bv, \bem, \bxi) \, \d \mu^{(t)} \;\longrightarrow\; \int g_i(\bv, \bem, \bxi) \, \d \mu = U_i(\bmpi).
    \]
    
    Recall that $g_i(\bv, \bem, \bxi) = \rho_i(\bem, \bxi) \, u_i(v_i)$. Since $u_i$ is continuous on the compact interval $I$, it must be bounded. Moreover, $\rho_i(\bem, \bxi) \in [0,1]$, so $g_i$ is bounded.
    
    The only potential discontinuities of $g_i$ arise from $\rho_i(\bem, \bxi)$, which depends on comparisons of the form $\xi_j \cdot m_j > \xi_\ell \cdot m_\ell$. We denote the set of ``tie points'' by
    \[
    T = \bigcup_{\substack{(j, \ell) \in [N]^2,\\ j \neq \ell}} \{ (\bem, \bxi) \mid \xi_j \cdot m_j = \xi_\ell \cdot m_\ell \}.
    \]
    For any $\bmpi$ and a pair of distinct agents $j, \ell$. Consider the random quadruple $(m_j, \xi_j, m_{\ell}, \xi_{\ell})$ drawn from the distribution $\pi_j \times \Xi \times \pi_{\ell} \times \Xi$. Condition on any fixed realization of $(m_j, m_\ell, \xi_\ell)$, the value of $\xi_j$ that would cause a tie $\xi_j \cdot m_j = \xi_\ell \cdot m_\ell$ is fixed and single. Since $\xi_j$ is drawn independently from the continuous uniform distribution $\Xi$ and all $m_j$ is at least $\underline{v} > 0$, we have both $\mu(T) = 0$ and $\mu^{(t)}(T) = 0$ for all $t$.
    
    On the other hand, for any point $(\bem, \bxi) \notin T$, the ranking of $(\xi_1 m_1, \ldots, \xi_N m_N)$ does not change locally within a small neighborhood of the point, so $\rho_i(\bem, \bxi)$ is locally constant and hence continuous. Therefore, $g_i$ is continuous except on a $\mu$- and $\mu^{(t)}$-null set $T$.
    
    Since $g_i$'s are uniformly bounded and continuous almost everywhere, and $\mu^{(t)} \Rightarrow \mu$ (weak convergence on the compact domain), the Portmanteau theorem \citep*[Theorem 2.7]{billingsley}\footnote{The exact theorem is using a pushforward notation $\mu^{(t)} g_i^{-1} \longrightarrow \mu g_i^{-1}$, which translates to the equation below.} implies
    \[
    \int g_i(\bv, \bem, \bxi) \, \d \mu^{(t)} \;\longrightarrow\; \int g_i(\bv, \bem, \bxi) \, \d \mu.
    \]
    Hence $U_i(\bmpi^{(t)}) \to U_i(\bmpi)$, proving that $U_i$ is continuous on $S_1 \times \cdots \times S_N$.
\end{proof}

So far, we have prove that the payoff function is continuous and the strategy space is a compact metric space, we utilize the following Glicksberg's Theorem to finish the proof of the existence of Nash equilibrium.

\begin{lemma}[Glicksberg's Theorem \citep*{glicksberg1952further}]\label{thm:Glicksberg}
A game with a finite number of players has a Nash equilibrium in mixed strategies if:
\begin{enumerate}
\item Each player's strategy space is a nonempty, compact, convex metric space.
\item Each player's utility function is continuous.
\end{enumerate}
\end{lemma}

By applying Glicksberg's Theorem with \cref{lem:noisy_compact} and \cref{lem:noisy_continuous}, we conclude that the Nash equilibrium exists for this game. 

\subsection{PoA Upper Bound}
Given that the Nash equilibrium always exists, we are going to prove the following PoA upper bound for the noisy game defined above.

\begin{theorem}\label{thm:noisy_main}
    The PoA of the $\eta$-noisy Bayesian persuasion game with $k$ selection (for any $k \ge 1$) is at most $(11 + 5 \sqrt{5}) \cdot \left(\frac{1 + \eta}{1 - \eta}\right)^2$.
\end{theorem}

The is almost identical to \cref{sec:general}. The construction of quantile cuts and the deviation signal $s^*$, the definition of winning probability and contribution all remain the same, and \cref{lem:swub} to \cref{lem:deviation_wp_lb} hold still. We only need to modify \cref{lem:klargest_win_lb} and \cref{lem:PoA_M_ub}. Denote the strategy profile of agents in the Nash Equilibrium by $\bpis$. We keep the rest of the notations the same as \cref{sec:general}, but now everything is based on the Nash equilibrium profile $\bpis$ of the noisy game.

\cref{lem:klargest_win_lb} need to be modified into the following lemma since the observation is noisy:
\begin{lemma} [Modified version of \cref{lem:klargest_win_lb}]\label{lem:klargest_win_lb_noisy}
If agent $i^*$ is selected with probability less than $\tau$ when $s^*$ is sent in $\bpis'$, we have \[\Pr\left[M(\bpis) \ge e'_{i^*}(s^*) \cdot \frac{1 - \eta}{1 + \eta} \right] \ge 1 - \tau.\] 
\end{lemma}
\begin{proof}
    If not, we have $\Pr\left[M(\bpis) \cdot (1 + \eta) < e'_{i^*}(s^*) \cdot (1 - \eta) \right] \ge \tau$. Denote by $M_{-i^*}(\bpis)$ the random variable representing the $(k - |\nn_\ell|)$-th largest posterior mean among all agents except $i^*$ in $\nn_s$ in $\bpis$. $M_{-i^*}(\bpis)$ is stochastically dominated by $M(\bpis)$. When $M_{-i^*}(\bpis) \cdot (1 + \eta) < e'_{i^*}(s^*) \cdot (1 - \eta)$, even with noisy observation, agent $i^*$ must be selected when $s^*$ is sent. This contradicts the assumption that she is selected with probability less than $\tau$ when $s^*$ is selected. 
\end{proof}

Similarly, \cref{lem:PoA_M_ub} is modified into the following:
\begin{lemma} \label{lem:PoA_M_ub_noisy}
     Given a strategy profile $\bpis$, if 
     $\Pr\left[M(\bpis) \ge e'_{i^*}(s^*) \cdot \frac{1 - \eta}{1 + \eta} \right] \ge 1 - \tau$,
     then we have 
     \[
     \SW(\bpis) \ge \SW' \cdot \frac{(1 - \tau) \cdot \left(\frac{1}{\phi} - \frac{\beta}{\tau \cdot \phi} \right) \cdot \left(\frac{1 - \eta}{1 + \eta}\right)^2}{2}.
     \]
\end{lemma}
\begin{proof}
    Similar to the proof of \cref{lem:PoA_M_ub}, we consider a hypothetical scenario where, instead of selecting the $k$ agents with the highest posterior means, the \principal selects all agents in $\nn_\ell$ and selects the remaining $k - |\nn_\ell|$ agents from $\nn_s$ with the highest posterior means \textbf{without noise}.  Denote the resulting social welfare in this scenario by $\SW^{\textrm{(fake)}}_{\textrm{(clear)}}(\bpis)$. We have
    \begin{align*}
        \SW^{\textrm{(fake)}}_{\textrm{(clear)}}(\bpis) &\ge  (k - |\nn_\ell| ) \cdot \E[M(\bpis)] + \sum_{i \in \nn_\ell} E_{i} \notag\\
        & \geq (1-\tau) \cdot (k - |\nn_\ell| ) \cdot  e'_{i^*}(s^*) \cdot \frac{1 - \eta}{1 + \eta} +  \sum_{i \in \nn_\ell} E_{i}  \tag{w.p. $1-\tau$, $M(\bpis) \ge e'_{i^*}(s^*) \cdot \frac{1 - \eta}{1 + \eta}$. }\\
        & \ge \Bigg[ (1 - \tau) \cdot (k - |\nn_\ell|) \cdot \left(\frac{1}{\phi} - \frac{\beta}{\tau \cdot \phi} \right) \cdot E_{i^*} + \sum_{i \in \nn_\ell} E_{i} \Bigg] \cdot \frac{1 - \eta}{1 + \eta} \tag{\cref{lem:ei}}\\
        & \ge \frac{(1 - \tau) \cdot (k - |\nn_\ell|) \cdot \left(\frac{1}{\phi} - \frac{\beta}{\tau \cdot \phi} \right) \cdot E_{i^*} + \sum_{i \in \nn_\ell} E_{i} }{(1 - \tau) \cdot \left(\frac{1}{\phi} - \frac{\beta}{\tau \cdot \phi} \right) \cdot (k - |\nn_\ell|) + |\nn_\ell|}\cdot (1 - \tau) \cdot \left(\frac{1}{\phi} - \frac{\beta}{\tau \cdot \phi} \right) \cdot k \cdot \frac{1 - \eta}{1 + \eta}.
    \end{align*}

    Note that everything except the last multiplicative factor $\frac{1 - \eta}{1 + \eta}$ is the same as the lower bound of $\SW^{\textrm{(fake)}}(\bpis)$ in \cref{eqn:fake_lb_1}. By applying the same argument in \cref{sec:general}, we have
    \begin{align*}
        \SW^{\textrm{(fake)}}_{\textrm{(clear)}}(\bpis) & \ge \frac{k \cdot E_\cut + \sum_{i \in \nn_s} E_i \cdot q_i + \sum_{i \in \nn_\ell} E_i}{k + \sum_{i \in \nn_s} q_i + |\nn_\ell|} \cdot (1 - \tau) \cdot \left(\frac{1}{\phi} - \frac{\beta}{\tau \cdot \phi} \right) \cdot k \cdot \frac{1 - \eta}{1 + \eta}\\
        & = \frac{\SW'}{2 k} \cdot (1 - \tau) \cdot \left(\frac{1}{\phi} - \frac{\beta}{\tau \cdot \phi} \right) \cdot k \cdot \frac{1 - \eta}{1 + \eta}. 
    \end{align*}
    
    Consider another hypothetical scenario where the agents are still playing the equilibrium strategies $\bpis$, but the receiver is observing the signals without noise and selects the $k$ agents with the highest posterior means. Denote the social welfare in this hypothetical scenario by $\SW_{\textrm{(clear)}}(\bpis)$. Since in the fake scenario, the \principal may not be a posterior mean maximizer, we have $\SW^{\textrm{(fake)}}_{\textrm{(clear)}}(\bpis) \le \SW_{\textrm{(clear)}}(\bpis)$. This completes the proof. Since the \principal is observing the values with noises\footnote{More detailedly, this is because the selected noisy values are at least $1/(1 + \eta)$ times the selected real values, and the selected noisy values are at least the noisy values of the actual largest $k$ agents, which is at least $(1 - \eta)$ times the real values of the actual largest $k$ agents.}, the real social welfare with noisy observation $\SW_{\textrm{(noisy)}}(\bpis)$ is at least $\frac{1 - \eta}{1 + \eta}$ fraction of the social welfare without noise. Therefore, we have 
    \begin{align*}
    \SW_{\textrm{(noisy)}}(\bpis) \ge \frac{1 - \eta}{1 + \eta} \cdot \SW_{\textrm{(clear)}}(\bpis) & \ge \frac{1 - \eta}{1 + \eta} \cdot \SW^{\textrm{(fake)}}_{\textrm{(clear)}}(\bpis) \\
    &\ge \frac{\SW'}{2} \cdot (1 - \tau) \cdot \left(\frac{1}{\phi} - \frac{\beta}{\tau \cdot \phi} \right) \cdot \left(\frac{1 - \eta}{1 + \eta}\right)^2. \qedhere
    \end{align*}
\end{proof} 
Therefore, we have the following upper bound of PoA for this noisy game:
\[
\mbox{PoA} \le \max \left\{\frac{2}{\beta \cdot (1 - 1 / \alpha)} , \frac{2}{(1 - \tau) \cdot \left(\frac{1}{\phi} - \frac{\beta}{\tau \cdot \phi} \right)} \cdot \left(\frac{1 + \eta}{1 - \eta}\right)^2\right\}.
\]

By setting the parameters exactly same as \cref{sec:general}, we get PoA of the noisy game is at most $(11+5\sqrt{5})\cdot \left(\frac{1 + \eta}{1 - \eta}\right)^2$, completing the proof of \cref{thm:noisy_main}.

\end{document}